\documentclass{article}
\usepackage[utf8]{inputenc}
\usepackage[english]{babel}
\usepackage[margin=1in]{geometry}

\usepackage{hyperref}       
\usepackage{url}            
\usepackage{booktabs}       
\usepackage{wrapfig}		

\usepackage{amsfonts}       
\usepackage{nicefrac}       
\usepackage{float}			
\usepackage{microtype}      
\usepackage{multirow}
\usepackage{amsmath, amssymb, amsthm}
\usepackage{comment}
\usepackage{caption}
\usepackage{subcaption}
\usepackage{graphicx}
\usepackage{color}
\usepackage{cite}
\usepackage{bbm}
\usepackage{enumitem,kantlipsum}
 
\DeclareMathOperator*{\argmax}{arg\,max}
\DeclareMathOperator*{\argmin}{arg\,min}

\newcommand{\E}{\mathcal{E}}
\newcommand{\e}{e}
\newcommand{\N}{N}
\newcommand{\Psa}{\Pi_{A}}
\newcommand{\Psb}{\Pi_{B}}
\newcommand{\ex}{\mathbb{E}}
\newcommand{\s}{\mathcal{S}}
\newcommand{\A}{\mathcal{A}}
\newcommand{\B}{\mathcal{B}}
\newcommand{\HH}{\mathbb{H}}

\newcommand{\T}{\mathcal{T}}
\newcommand{\PP}{\mathbb{P}}
\newcommand{\I}{\mathbb{I}}
\newcommand{\D}{\mathcal{D}}
\newcommand{\DD}{\mathbb{D}}
\newcommand{\te}{\tilde{\mathcal{E}} }

\newcommand{\shk}{s_h^k}
\newcommand{\ahk}{a_h^k}
\newcommand{\bhk}{b_h^k}
\newcommand{\rhk}{r_h^k}
\newcommand{\shhk}{s_{h+1}^k}
\newcommand{\brmu}{\nu^{\dagger}(\mu)}
\newcommand{\brnu}{\mu^{\dagger}(\nu)}
\newcommand{\reg}{\mathsf{Reg}}
\newcommand{\mui}{\mu_{\text{IDS}}}
\newcommand{\nui}{\nu_{\text{IDS}}}
\newcommand{\muids}{\mu_{\text{IDS}}^k}
\newcommand{\muri}{\mu_{\text{R-IDS}}^k}
\newcommand{\murii}{\mu_{\text{R-IDS}}}
\newcommand{\nurii}{\nu_{\text{R-IDS}}}
\newcommand{\muci}{\mu_{\text{C-IDS}}^k}
\newcommand{\nuci}{\nu_{\text{C-IDS}}^k}
\newcommand{\mucii}{\mu_{\text{C-IDS}}}
\newcommand{\nucii}{\nu_{\text{C-IDS}}}

\newcommand{\pits}{\pi_{\text{TS}}^{(i)}}

\newcommand{\pids}{\pi_{\text{R-IDS}}}

\newcommand{\pidskne}{\pi^{\text{NE},k}_{\text{R-IDS}}}
\newcommand{\pidskcce}{\pi^{\text{CCE},k}_{\text{R-IDS}}}

\newcommand{\nuri}{\nu_{\text{R-IDS}}^k}
\newcommand{\nuids}{\nu_{\text{IDS}}^k}
\newcommand{\muts}{\mu_{\text{TS}}^k}
\newcommand{\nuts}{\nu_{\text{TS}}^k}
\newcommand{\me}{\bar{e}_k}
\newcommand{\met}{\tilde{e}_k}
\newcommand{\meh}{\hat{e}_k}

\newcommand{\tl}{\widetilde{\lambda}}
\newcommand{\trh}{r^{\me}_h}
\newcommand{\rhdd}{r_h^{\me'}}
\newcommand{\C}{C_{\epsilon}}
\newcommand{\Pipi}{\Pi_{i}^{\text{pure}}}

\newcommand{\mc}[1]{{\color{black}#1}}

\newtheorem{theorem}{Theorem}
\newtheorem{lemma}{Lemma}
\newtheorem{remark}{Remark}
\newtheorem{definition}{Definition}
\newtheorem{example}{Example}
\newtheorem{corollary}{Corollary}

\title{Provably Efficient Information-Directed Sampling Algorithms for Multi-Agent Reinforcement Learning}

\author{Qiaosheng Zhang\thanks{Shanghai Artificial Intelligence Laboratory} \ \ 
	Chenjia Bai\footnotemark[1] \ \ Shuyue Hu\footnotemark[1] \ \ Zhen Wang\thanks{Northwestern Polytechnical University} \ \ Xuelong Li\thanks{Institute of Artificial Intelligence (TeleAI), China Telecom Corp Ltd} 
	} 

\begin{document}

\maketitle

\begin{abstract}

This work designs and analyzes a novel set of algorithms for multi-agent reinforcement learning (MARL) based on the principle of information-directed sampling (IDS). These algorithms draw inspiration from foundational concepts in information theory, and are proven to be sample efficient in  MARL settings such as two-player zero-sum Markov games (MGs) and multi-player general-sum MGs. For episodic two-player zero-sum MGs, we present three sample-efficient algorithms for learning Nash equilibrium. The basic algorithm, referred to as {\sc MAIDS}, employs an asymmetric learning structure where the max-player first solves a minimax optimization problem based on the \emph{joint information ratio} of the joint policy, and the min-player then minimizes the \emph{marginal information ratio} with the max-player’s policy fixed. Theoretical analyses show that it achieves a Bayesian regret of $\tilde{O}(\sqrt{K})$ for $K$ episodes. To reduce the computational load of {\sc MAIDS}, we develop an improved algorithm called {\sc Reg-MAIDS}, which has the same Bayesian regret bound while enjoying less computational complexity. Moreover, by leveraging the flexibility of IDS principle in choosing the learning target, we propose two methods for constructing compressed environments based on rate-distortion theory, upon which we develop an algorithm {\sc Compressed-MAIDS} wherein the learning target is a compressed environment. Finally, we extend {\sc Reg-MAIDS} to multi-player general-sum MGs and prove that it can learn either the Nash equilibrium or coarse correlated equilibrium in a sample efficient manner.

\end{abstract}

\section{Introduction}

The problem of multi-agent reinforcement learning (MARL), where multiple agents  learn and make sequential decisions in a shared environment to optimize their individual or collective rewards, has become increasingly relevant in real-world applications such as robot systems~\cite{brambilla2013swarm}, autonomous driving~\cite{shalev2016safe}, multi-player games~\cite{silver2016mastering}, etc. A crucial consideration in MARL is the \emph{sample efficiency}, as it directly influences the practical applicability and scalability of MARL algorithms in real-world multi-agent systems. The pursuit of sample efficient algorithms stands as a critical and pressing concern in the realm of MARL.

In recent years, there have been efforts to develop sample-efficient MARL algorithms, and some accompanied with theoretical guarantees. Most existing provably efficient algorithms are based on the principle of \emph{optimism in the face of uncertainty} (OFU). Notable OFU-based algorithms include Nash Q-learning~\cite{bai2020near}, Nash V-learning~\cite{jin2021v}, and model-based algorithms~\cite{huang2021towards, liu2021sharp}. Another principled but less explored algorithm design paradigm for MARL is \emph{posterior sampling}. Very recently, two studies demonstrated that posterior sampling-based algorithms can be both sample efficient and computationally efficient in MARL, either with full observation~\cite{xiong2022self} or partial observation~\cite{qiu2023posterior}. 

In the general paradigm of posterior sampling, \emph{information-directed sampling} (IDS) stands out as a relatively new yet principled exploration strategy for sequential decision-making problems~\cite{russo2014learning,russo2018learning,hao2022regret}. Drawing inspiration from information theory, IDS tackles the exploration-exploitation tradeoff by requiring the agent to balance the policy's sub-optimality (exploitation) and the acquired information about a \emph{learning target} (exploration) via a quantity called \emph{information ratio}.  Compared with OFU-based and other posterior sampling algorithms (such as Thompson sampling~\cite{russo2018tutorial}), IDS has several appealing features. For instance, (i) it offers flexibility in choosing which kind of information to gain/learn (referred to as the \emph{learning target}) when encountering different information structures;\footnote{This feature is particularly beneficial when the environment is overwhelmingly complex or when it consists of specific information structures that can be compressed. Example~\ref{example:compress} illustrates how to leverage this feature to improve sample efficiency.} and (ii) IDS is able to accommodate parametric uncertainty and heteroscedastic observation noise (e.g., when states can only be partially observed)~\cite{nikolov2018information}. As a result, IDS outperforms Thompson sampling in terms of sample efficiency in a variety of settings~\cite{lu2023reinforcement}.   Empirically, IDS-based algorithms have also demonstrated superior performance and computation efficiency in bandit problems~\cite{kirschner2018information, hao2021information} and reinforcement learning (RL) problems~\cite{nikolov2018information, hao2023exploration}, affirming its practicality in real-world scenarios.   

Despite the theoretical advantages and empirical success, it remains unknown if the IDS principle can be applied to \emph{competitive} or \emph{cooperative} multi-player decision-making problems. If it can, the question that arises is whether IDS-based algorithms can maintain favorable features  that have been demonstrated in bandit and RL problems. To address these unexplored questions, this work initiates an investigation into the potential of IDS in MARL settings, specifically focusing on the episodic two-player zero-sum {Markov game} (MG) and multi-player general-sum MG. We put forth the first line of MARL algorithms designed based on the IDS principle. One attractive feature of our algorithm is that when the environment (or transition kernel) admits a compressed approximation, such approximation can be utilized in constructing the learning target in our algorithm, resulting in an improved regret bound. The concept of compression is inspired by the classical lossy compression problems (a.k.a. rate-distortion theory) in information theory~\cite{berger2003rate} and other related works in the RL literature~\cite{arumugam2021deciding, arumugam2021value, arumugam2022deciding, bai2024pessimistic}.



\paragraph{Main contributions} This work introduces new IDS-based MARL algorithms that are driven by information theory. These algorithms enrich the set of sample-efficient MARL algorithms that are theoretically known to satisfy a $\tilde{O}(\sqrt{K})$  regret bound\footnote{We say $f(n) = \tilde{O}(g(n))$ if $f(n) = O(g(n) \cdot \text{polylog}(n))$.} for MGs with $K$ episodes. We detail our contributions as follows.
\begin{enumerate}
\item We first develop a basic self-play algorithm, named {\sc MAIDS}, for learning Nash equilibrium (NE) in two-player zero-sum MGs. It operates by letting players sequentially optimize the \emph{joint information ratio} and   \emph{marginal information ratio} at each episode. These ratios represent the expected regret over the acquired information about a learning target, and we choose the learning target to be the entire environment in {\sc MAIDS}. Theorem~\ref{thm:maids} presents a Bayesian regret bound that scales as $\tilde{O}(\sqrt{K})$ for $K$ episodes of MGs, and  is valid for all prior distributions of the environment. 

\item We also develop an  algorithm {\sc Reg-MAIDS} that offers reduced computational complexity compared to {\sc MAIDS}, without compromising the sample efficiency (Theorem~\ref{thm:proof_thm2}). {\sc Reg-MAIDS}  can be implemented efficiently by leveraging existing computationally efficient MARL algorithms. 

\item Given the flexibility of the IDS principle in selecting the learning target, we have developed an algorithm named {\sc Compressed-MAIDS}, where the learning target is a \emph{compressed environment} (instead of the entire environment in {\sc MAIDS}).  Inspired by lossy compression in information theory, we introduce two principles for constructing the compressed environment, and provide a Bayesian regret bound for {\sc Compressed-MAIDS} under a specific compressed environment (Theorem~\ref{thm:compressed}).

\item Finally, we extend {\sc Reg-MAIDS} to multi-player general-sum MGs, and show that it can learn either the NE or  coarse correlated equilibrium (CCE) sample-efficiently through the derivation of Bayesian regret bounds (Theorem~\ref{thm:general}).

\end{enumerate}

It is worth noting that while MGs can be viewed as extensions of Markov decision process (MDP) for RL, the unpredictability of other players' actions (called \emph{non-stationary} in~\cite{zhang2021multi}) imposes new challenges in algorithm design and in the construction of compressed environments. For instance, to tackle
the competitive nature of zero-sum MGs, we adopt an \emph{asymmetric learning procedure} in the design of IDS-based algorithms. Specifically,  in order to learn an approximate NE policy for the max-player, the max-player first chooses a policy that optimizes the joint information ratio against a (fictitious) \emph{worst-case opponent}, and the min-player subsequently optimizes the marginal information ratio to \emph{assist} max-player's learning. \mc{This two-step procedure differs notably from previous IDS algorithms for bandit and  single-agent RL settings. Additionally, our analytical techniques, which incorporate information-theoretic methods, also differ from the techniques for OFU-based and Thompson sampling algorithms in the context of MARL.}

\paragraph{Related works} 
Driven by the successful applications of MARL techniques in real-world applications, there has been a growing focus on the theoretical exploration of MARL in recent years. A number of studies have put forth sample-efficient algorithms for various representative MARL settings, including the tabular zero-sum MG~\cite{bai2020provable,bai2020near,qiu2021provably,liu2021sharp}, zero-sum MG with linear function approximation~\cite{xie2020learning,chen2022almost} and with general function approximation~\cite{jin2021power, huang2021towards, xiong2022self,qiu2023posterior,liu2023maximize, mao2023role}, as well as  general-sum MGs~\cite{ jin2020sample,hu2019modelling,song2021can,hu2022dynamics, liu2022sample,foster2023complexity,hu2023best,wang2023breaking,cui2023breaking,xiong2023sample}. These works provides effective, efficient and theoretically-sound solutions  for finding the (approximate) Nash Equilibrium of MGs, based on  the principle of OFU or posterior sampling. In contrast, our proposed algorithms are founded on the IDS principle.  

The IDS principle was first proposed by~\cite{russo2014learning, russo2018learning} for bandit settings, and has since been explored in both the bandit and broader RL settings~\cite{nikolov2018information,liu2018graph,kirschner2021asymptotically, hao2021information, lu2023reinforcement}. Theoretical analyses of IDS in the RL setting are credited to~\cite{lu2019confidence} and~\cite{hao2022regret}. The former develops a Bayesian regret bound when the prior distribution of the environment is specialized to the Dirichlet distribution, while the latter establishes prior-free Bayesian regret bounds for the first time. Although the regret bounds of the algorithms in~\cite{hao2022regret} does not match the information-theoretic lower bound, we notice that, with the help of a new analytical technique for Thompson sampling~\cite{moradipari2023improved}, closing this gap becomes possible. Our work is closely related to~\cite{hao2022regret}, as the algorithm design and proof technique for {\sc MAIDS} and {\sc Reg-MAIDS} are inspired by their work. However, substantial efforts are required to address the non-stationary nature of MGs. Furthermore, the construction of the compressed environment, as well as the design and analyses of {\sc Compressed-MAIDS}, are significantly different from theirs.    Besides,~\cite{chakraborty2023steering} proposes a variant of IDS-based RL algorithm based on the use of \emph{Stein information}, and demonstrates the advantage of computational efficiency of their algorithm through both theoretical and experimental analyses. 

As the ultimate goal of MARL is to make good decisions rather than to estimate/learn the environment, it is often not necessary to learn/explore every granular detail of the environment. When the environment comprises redundant information that is not helpful for making decision, it can be counterproductive for agents to excessively focus on exploring/learning the redundant information rather than to exploit.      This motivates the consideration of compressed environments, and coincidentally, the IDS principle allows for the selection of a compressed environment as the learning target. In recent RL literature,~\cite{hao2022regret} proposes a novel approach for constructing a compressed environment based on the partition of transition kernels and value functions, whereas the subsequent work~\cite{moradipari2023improved} provides a similar but refined construction. However, their approaches cannot be directly applied to MARL due to the effect of the opponent's unpredictable policy on the value functions.   Meanwhile,~\cite{arumugam2022deciding} proposes a construction of the compressed environment based on the rate-distortion theory, where the distortion measure is defined through  the \emph{value equivalence principle}~\cite{grimm2020value,grimm2021proper}. It is worth mentioning that our soft-compression principle is similar to their construction, but the accompanying distortion measure differs. A detailed comparison between the distortion measures is provided in Remark~\ref{remark:distortion}, Section~\ref{sec:compressed}.

\underline{Connections to DEC:} A concept that is related to IDS is the \emph{decision-estimation coefficient} (DEC), which is a complexity measure for sequential decision-making~\cite{foster2021statistical}. The DEC can be decomposed into two terms: one that represents the regret and another that quantifies the cumulative estimation error. The cumulative estimation error is measured by the squared Hellinger distance between the trajectories induced by the true model and induced by the estimated model. In contrast to the estimation error, our work focuses on how much information can be acquired through the interaction with the environment, where the amount of information is measured by a mutual information term.  It is worth noting that there is a similarity between our mutual information term and their estimation error term, to some extent. Specifically, mutual information is equivalent to the KL-divergence between distributions of trajectories, while their estimation error is measured by the squared Hellinger distance between distributions of trajectories.

Besides, a recent work~\cite{foster2023complexity} generalizes the concept of DEC to multi-agent settings, and proposes a complexity measure called \emph{multi-agent decision-estimation coefficient} (MA-DEC) for multi-agent decision making problems. Our work differs from~\cite{foster2023complexity} in three aspects: (i) They consider frequentist settings while we consider Bayesian settings; (ii) The MA-DEC depends on the estimation error about the environment, while the information ratio in our work depends on the acquired information about the environment (or about the compressed environment); (iii) They focus on the statistical complexity while we additionally provide a computationally efficient algorithm (i.e., the {\sc Reg-MAIDS}).

\underline{Connections to AIR:} The IDS principle is also related to the concept of \emph{algorithmic information ratio} (AIR)~\cite{xu2023bayesian}. Their work establishes a theory for analyzing frequentist regrets through Bayesian-type algorithms in sequential decision-making problems.  A central object to be optimized in their work is the AIR, which, in addition to encompassing the terms of expected regret and acquired information involved in the information ratio term, further integrates a reference distribution to be aligned with. However, to the best of our knowledge, investigations into AIR are restricted to the simpler bandit and RL settings, while their applicability in the competitive MARL environment remains unknown.

\paragraph{Outline} The paper is organized as follows. Preliminaries on notational conventions and related information-theoretic concepts are introduced in Section~\ref{sec:pre}.  The mathematical formulation of the two-player zero-sum MG and the learning objectives are introduced in Section~\ref{sec:model}. We then present our sample-efficient algorithms {\sc MAIDS}, {\sc Reg-MAIDS}, and {\sc Compressed-MAIDS} (designed for zero-sum MGs) in Sections~\ref{sec:vanilla}-\ref{sec:compressed}, respectively. In Section~\ref{sec:general}, we introduce the mathematical formulation of the multi-player general-sum MG as well as our accompanying {\sc Reg-MAIDS} algorithm. Section~\ref{sec:conclusion} concludes this work and proposes future research directions. Most detailed proofs are provided in the Appendix.

\section{Preliminaries} \label{sec:pre}

\paragraph{Notations} 
For any positive integer $n \in \mathbb{N}^+$, we denote the set of positive integers ranging from $1$ to $n$ by $[n] \triangleq \{1,2,\ldots, n\}$. For any set $\mathcal{X}$, let $\Delta(\mathcal{X})$ be the probability simplex over $\mathcal{X}$ (i.e., the set of all possible probability distributions on $\mathcal{X}$). For any probability measures  $P$ and $Q$ on a same measurable space $\mathcal{X}$, we define their \emph{KL-divergence} (a.k.a. \emph{relative entropy}) as $$\DD_{\text{KL}}(P\Vert Q) \triangleq \int_{\mathcal{X}} \log (P(dx)/Q(dx)) P(dx)$$ if $P$ is absolutely continuous with respect to $Q$, where $P(dx)/Q(dx)$ is the Radon–Nikodym derivative of $P$ with respect to $Q$. 

We adopt \emph{asymptotic notations}, including $O(.)$, $o(.)$, $\Omega(.)$, $\omega(.)$, and $\Theta(.)$, to describe the limiting behaviour of functions/sequences. For instance, we say a pair of functions $f(n)$ and $g(n)$ satisfies $f(n) = O(g(n))$ if there exist $m > 0$ and $N_0 \in \mathbb{N}^+$ such that for all $n > N_0$, $|f(n)| \le mg(n)$. Moreover, we say $f(n) = \tilde{O}(g(n))$ if $f(n) = O(g(n) \cdot \text{polylog}(n))$.

\paragraph{Information Theory Preliminaries} 
For a pair of random variables $X$ and $Y$, we define their \emph{mutual information} as 
\begin{align}
\I(X;Y) \triangleq \DD_{\text{KL}}(\PP((X,Y) \in \cdot \ ) \Vert \PP(X \in \cdot \ ) \times \PP(Y \in \cdot \ )), \notag
\end{align}
which is also equivalent to the form $\ex_X [\DD_{\text{KL}}(\PP(Y \in \cdot \ |X) \Vert \PP(Y \in \cdot \ )].$
When introducing another random variable $Z$, one can define the \emph{conditional mutual information} as
$$\I(X;Y|Z) \triangleq \ex_Z[\DD_{\text{KL}}(\PP((X,Y) \in \cdot \  |Z) \Vert \PP(X \in \cdot \ |Z) \times \PP(Y \in \cdot \  |Z))].$$
For a collection of random variables $(X_0, X_1, X_2, \ldots, X_n)$, the \emph{chain rule} of mutual information~\cite{cover1999elements} states that 
$$\I(X_0; X_1, X_2, \ldots, X_n) = \sum_{i=1}^n \I(X_0;X_i|X_1, \ldots, X_{i-1}).$$

\section{Zero-Sum Markov Games} \label{sec:model}
In this section, we first introduce the mathematical formulation of the zero-sum MG that comprises two competitive players. The formulation of the more general multi-player general-sum MG is provided in Section~\ref{sec:general}.

The two-player zero-sum MG is denoted by  $\E = (H,\s,\A,\B, \{P_h\}_{h=1}^H, \{r_h\}_{h=1}^H )$, where $H$ is the length of each episode, $\s$ is the set of countable state space with cardinality $|\s| = S$, while $\A$ and $\B$ are the sets of action spaces of the max-player and min-player respectively, with cardinalities $|\A| = A$ and $|\B| = B$. For each step $h \in [H]$, $P_h: \s \times \A \times \B \to \Delta(\s)$ is the \emph{transition kernel} from the current state and actions to the next state. We use $r_h: \s \times \A \times \B \to [0,1]$ to denote the deterministic reward function. Without loss of generality, it is assumed that $\s$, $\A$, $\B$, $\{r_h\}_{h=1}^H$ are known\footnote{Extending to unknown and stochastic reward functions will not pose significant challenges, as learning transition kernels is more difficult than learning reward functions.} while the transition kernels $\{P_h\}_{h=1}^H$ are unknown and random. We also refer to $\E$ as the  \emph{environment} of the MARL problem.   

\paragraph{Prior distributions} We consider a Bayesian setting where we have a prior belief on the environment $\E$, or equivalently, on the transition kernel $\{P_h\}_{h=1}^H$, since the other model parameters $\s, \A, \B$ and $\{ r_h\}$ are assumed to be deterministic. For each step $h \in [H]$, let $\Theta_h$ be the parameter space of $P_h$, and let $\rho_h$ be the \emph{prior distribution} of $P_h$ on $\Theta_h$. Let $\Theta \triangleq \prod_{h=1}^H \Theta_h$ be the parameter space of the kernels $\{P_h\}_{h=1}^H$, and without loss of generality we assume $\Theta$ is convex. Let $\rho \triangleq \prod_{h=1}^H \rho_h$ be the \emph{product} prior distribution of $\{P_h\}_{h=1}^H$ on $\Theta$. The players know the prior distribution $\rho$ but not the realization of $\{P_h\}_{h=1}^H$. Note that the environment $\E$ can be viewed as a random variable with distribution $\rho$ since its only randomness comes from $\{P_h\}_{h=1}^H$.

\paragraph{Interaction Processes} In the zero-sum MG, each episode $k \in [K]$ starts at an initial state $s_1^k$. At each step $h \in [H]$, both the max-player and min-player observe the current state $\shk$, and pick their own actions $\ahk$ and $\bhk$ \emph{simultaneously}. The max-player receives a reward $\rhk = r_h(\shk,\ahk,\bhk)$, while the min-player receives $-\rhk$. The environment then transits to the next state $\shhk$ according to the transition kernel $P_h(\cdot|\shk,\ahk,\bhk)$. The episode ends when the final state $s_{H+1}^k$ is reached. Without loss of generality and for simplicity, we assume that the initial state of each episode $s_1^k$ is fixed to the state $s_1$ over episodes. Generalizing such an assumption to having a random initial state with a fixed but known distribution will not pose significant challenges.    


\paragraph{Policies}  

A policy $\mu$ of the max-player is a collection of mappings $(\mu_1, \ldots, \mu_H)$ such that $\mu_h: \Omega_{h-1} \times \s \to \Delta(\A)$. Note that in the MG with simultaneous moves, the max-player cannot observe the min-player's action when choosing her/his own action. The policy of the min-player is a collection of mappings $\nu = (\nu_1, \ldots, \nu_H)$ such that $\nu_{h}: \Omega_{h-1} \times \s \to \Delta(\B)$.  Moreover, we say a policy is a  \emph{Markov policy} if each mapping $\mu_h: \s \to \Delta(\A)$ (or $\nu_h: \s \to \Delta(\B)$) only takes the current state as the input; that is, the action only depends on the current state but not the past trajectory. We denote the set of all Markov policies of the max-player by $\Psa$, and that of the min-player by $\Psb$. Without loss of generality, we only consider Markov policies in the remaining parts of the paper.
\paragraph{Value functions}
The goal of the max-player, as the name suggests, is to find a policy that maximizes the cumulative reward, while the goal of the min-player is to minimize the cumulative reward. The cumulative reward can be represented by the \emph{value function} and the \emph{action-value function} introduced below. The value function $V^{\E}_{h,\mu,\nu}: \s \to \mathbb{R}$ at step $h$ with respect to the environment $\E$ and the policy $(\mu,\nu)$ is defined as
\begin{align}
V^{\E}_{h,\mu,\nu}(s) \triangleq \ex_{\mu,\nu}^{\E}\left[ \sum_{h'=h}^H r_{h'}(s_{h'},a_{h'},b_{h'}) \big| s_h = s \right],    \quad \forall s \in \s,
\end{align}
where $\ex_{\mu,\nu}^{\E}$ denotes the expectation over the trajectory $\{s_h',a_h',b_h'\}_{h'=h}^H$ that is generated by the interaction between the policy $(\mu,\nu)$ and environment $\E$. The value function $V^{\E}_{h,\mu,\nu}(s)$ represents the expected cumulative reward from state $s$ at step $h$ when executing policy $(\mu,\nu)$. We also define the action-value function $Q^{\E}_{h,\mu,\nu}: \s \times \A \times \B \to \mathbb{R}$ at step $h$ with respect to the environment $\E$ and the policy $(\mu,\nu)$ as
\begin{align}
Q^{\E}_{h,\mu,\nu}(s,a,b) \triangleq \ex_{\mu,\nu}^{\E} \left[ \sum_{h'=h}^H r_{h'}(s_{h'},a_{h'},b_{h'}) \big| s_h = s, a_h = a, b_h = b \right],    \ \ \forall (s,a,b) \in \s\times \A\times \B.
\end{align}
The action-value function $Q^{\E}_{h,\mu,\nu}(s,a,b)$ represents the expected cumulative reward from state $s$ at step $h$ when  executing actions $(a_h,b_h) = (a,b)$ at step $h$ and policy $(\mu,\nu)$ afterwards. Moreover, we have the following Bellman equations:
\begin{align}
&Q^{\E}_{h,\mu,\nu}(s,a,b) = r_h(s,a,b) +\ex_{s' \sim P_h(\cdot|s,a,b)} [V_{h+1,\mu,\nu}^{\E} (s')],    \\
&V_{h,\mu,\nu}^{\E} (s) = \ex_{a \sim \mu_h(\cdot|s), b \sim \nu_h(\cdot|s) } [Q^{\E}_{h,\mu,\nu}(s,a,b)].
\end{align}

\paragraph{Best responses}

For any Markov policy $\mu$ of the max-player, there exists a \emph{best response} $\brmu$ from the min-player, which is a Markov policy  that satisfies $V^{\E}_{h,\mu,\brmu}(s) = \inf_{\nu} V^{\E}_{h,\mu,\nu}(s)$ for all $(s,h)$. Similarly, if the min-player's policy $\nu$ is given, the max-player can also find a best response $\brnu$ that satisfies $V^{\E}_{h,\brnu,\nu}(s) = \sup_{\mu} V^{\E}_{h,\mu,\nu}(s)$ for all $(s,h)$. For notational convenience, we introduce the following abbreviations of value functions:
\begin{align}
   V^{\E}_{h,\mu,\dagger}(s) \triangleq V^{\E}_{h,\mu,\brmu}(s), \quad \text{and} \quad V^{\E}_{h,\dagger,\nu}(s) \triangleq V^{\E}_{h,\brnu,\nu}(s).
\end{align}

\paragraph{Nash Equilibrium} It is well known~\cite{filar2012competitive} that for a MG represented by $\E$, there exists a \emph{Nash Equilibrium (NE) policy} $(\mu^*, \nu^*)$, where both $\mu^*$ and $\nu^*$ are Markov policies, that satisfies
\begin{align}
 V^{\E}_{h,\mu^*, \dagger}(s) =  \sup_{\mu} V^{\E}_{h,\mu, \dagger}(s) \quad \text{and} \quad V^{\E}_{h, \dagger, \nu^*}(s) =  \inf_{\nu} V^{\E}_{h, \dagger,\nu}(s), \quad \forall (s,h) \in \s \times [H].
\end{align}
This means that the policy $\mu^*$ (or policy $\nu^*$) is optimal when the opponent can always choose the best response, i.e., it can be regarded as ``the best response to the best response''. The NE policy $(\mu^*,\nu^*)$  also satisfies the minimax equation:
\begin{align}
    \sup_{\mu} \inf_{\nu} V^{\E}_{h,\mu, \nu}(s) =  V^{\E}_{h,\mu^*, \nu^*}(s) = \inf_{\nu} \sup_{\mu} V^{\E}_{h,\mu, \nu}(s) \quad \forall (s,h) \in \s \times [H]. 
\end{align}
Moreover, for each player, there is no incentive to move away from its NE policy if the other player does not move, in the sense that
\begin{align}
   V^{\E}_{h,\mu, \nu^*}(s) \le V^{\E}_{h,\mu^*, \nu^*}(s) \le V^{\E}_{h,\mu^*, \nu}(s) \label{eq:6}
\end{align}
for any policies $\mu$ and $\nu$ and for all $(s,h) \in \s \times [H]$. We sometimes write $(\mu^*,\nu^*)$ as $(\mu^*(\E),\nu^*(\E))$ and $V^{\E}_{h,\mu^*, \nu^*}(s)$ as $V^{\E}_{h,\mu^*(\E), \nu^*(\E)}(s)$  to explicitly highlight the fact that the Nash policy $(\mu^*,\nu^*)$ depends on the environment $\E$.    

It is also known that the NE policy may not be unique; however, the corresponding values are the same. Thus, one can abbreviate the value  $V^{\E}_{h,\mu^*, \nu^*}(s)$ of any NE policy $(\mu^*,\nu^*)$ as $V^{\E,*}_{h}(s)$ for simplicity. We refer to $V^{\E,*}_{h}(s)$ as the \emph{Nash value}.

\paragraph{Learning objectives} We first focus on the max-player. The goal of the max-player is to learn a policy $\mu$ that is almost as good as the NE policy $\mu^*$, in the sense that the corresponding value $V_{1,\mu,\dagger}^{\E}(s_1)$ against the best response of the min-player is close to the value  $V_{1,\mu^*,\dagger}^{\E}(s_1)$. Note that $V_{1,\mu^*,\dagger}^{\E}(s_1)$ is  equal to  $V_{1}^{\E,*}(s_1)$ by the definition of NE policy. Thus, an appropriate way to define the regret of the max-player is through the difference between $V_{1}^{\E,*}(s_1)$ and $V_{1,\mu,\dagger}^{\E}(s_1)$. 

Suppose the two players interact with the environment for $K \in \mathbb{N}^+$ episodes.  Let $\mu = \{\mu^k \}_{k=1}^K$ be the max-player's policy, where $\mu_k$ is the policy for episode $k$. For a fixed realization of the environment\footnote{We prefer using the uppercase letter $\E$ to denote the random variable while using the lowercase letter $e$ to denote a realization.} $\E = \e$, the cumulative regret over $K$ episodes is defined as
\begin{align}
\reg_K(\e, \mu) \triangleq \sum_{k=1}^K V_1^{e,*}(s_1) - V_{1,\mu_k, \dagger}^e (s_1).
\end{align}
The Bayesian regret is defined as 
\begin{align}
\mathsf{BR}_K (\mu) \triangleq  \ex \left(\reg_K(\E,\mu) \right) = \ex \left(\sum_{k=1}^K V_1^{\E,*}(s_1) - V_{1,\mu_k, \dagger}^{\E} (s_1) \right), 
\end{align}
where the expectation is over the randomness in the environment $\E$ (or equivalently, over the prior probability measure $\rho$ of the transition kernels $\{P_h\}_{h=1}^H$ ). Note that one can also define the regret $\reg_K(\e, \nu)$ and Bayesian regret $\mathsf{BR}_K (\nu)$ for the min-player in a symmetric fashion. \mc{Adding the (Bayesian) regrets of the two players together yields the classical duality gap $\sum_k (V_{1,\dagger, \nu^k}^e -  V_{1,\mu^k, \dagger}^e)$ and its Bayesian version. A zero duality gap implies NE.} 

\begin{table}[h!]
\centering
\caption{Notations for zero-sum MGs} \vspace{3pt}
\begin{tabular}{ p{3cm}  p{12cm}  }
 \hline\hline
 $\Delta(\mathcal{X})$   & Probability simplex over the set $\mathcal{X}$ \\
 $\s$   & State space \\
 $\A, A$&   Action space of the max-player, cardinality of $\A$  \\
 $\B, B$&   Action space of the min-player, cardinality of $\B$ \\
 $H$&   Length of each episode \\
 $P_h$&   Transition kernel at step $h$ ($P_h: \s \times \A \times \B \to \Delta(\s)$) \\
 $r_h$&   Reward function at step $h$ ($r_h: \s \times \A \times \B \to [0,1]$) \\
 $\E$ & Environment (random variable) \\
 $e$ & Environment (realization) \\
 $K$ & Number of episodes \\
 $\Theta$ & Parameter space of the transition kernels $\{P_h\}_{h=1}^H$ \\
 $\rho$ & Prior distribution of the kernels $\{P_h\}_{h=1}^H$ \\
 $\T_h^k$ & Trajectory at episode $k$ up
to step $h$ \\ 
  $\D_k$ & Full trajectory up to the beginning of episode $k$ \\
  $\Pi_A/\Pi_B$ & Set of all Markov policies of the max-player/min-player \\
  $V_{h,\mu,\nu}^{\E}(s)$ & Value of state $s$ at step $h$ (w.r.t. policy $(\mu,\nu)$ and environment $\E$) \\
  $\mu^\dagger(\nu)$ & Best response to policy $\nu$ \\
  $V_{h,\dagger,\nu}^{\E}(s)$ & Value of state $s$ w.r.t. $\nu$ and $\mu^{\dagger}(\nu)$  (i.e., $V_{h,\dagger,\nu}^{\E}(s) 
 = V_{h,\mu^\dagger(\nu),\nu}^{\E}(s)$)\\
  $(\mu^*, \nu^*)$ & Nash equilibrium policy \\
  $V_h^{\E,*}(s)$ & Nash value of state $s$ at step $h$ w.r.t. environment $\E$ \\
  $\chi$ & Learning target \\
  $\ex_k(\cdot)$ & Expectation conditioned on the past trajectory (i.e., $\ex_k(\cdot) = \ex(\cdot | \D_k)$) \\
  $\I_k(\cdot; \cdot)$ & Mutual information conditioned on $\D_k$ \\
  $\Gamma(\mu,\nu,\chi)$ & Information ratio  \\
  $\Lambda_k^{\mu}(\nu, \chi)$ & Conditional information ratio w.r.t. policy $\mu$ \\
  $\me, \me'$ & Mean environments defined in Section~\ref{sec:regularized} \\
  $\Phi_A/\Phi_B$ & Subset of Markov policies of the max-player/min-player \\
  $d_{\Phi_A, \Phi_B}(\cdot,\cdot)$ & Distortion measure w.r.t. $\Phi_A$ and $\Phi_B$ \\
  $\te$ & Compressed environment \\
  $\mathcal{C}_{\delta}$ &  $\delta$-covering of probability simplex $\Delta(\s)$ \\
  $\kappa(\delta)$ & $\delta$-covering number of probability simplex $\Delta(\s)$ \\
  
 \hline\hline
\end{tabular}
\label{table:1}
\end{table}

\section{The  Basic Algorithm: {\sc MAIDS}} \label{sec:vanilla}

This section considers  zero-sum MGs and presents a basic multi-agent version of IDS-based algorithm, named {\sc MAIDS}, with theoretical guarantees on Bayesian regrets. 

\mc{Let $\T_{H+1}^k \triangleq \{s_1^k, a_1^k, b_1^k, r_1^k, \ldots, s_H^k, a_H^k, b_H^k, r_H^k, s_{H+1}^k\}$ be the \emph{trajectory} of episode $k$, and $\D_k \triangleq \{\T^i_{H+1}\}_{i=1}^{k-1} $ be the full trajectory up to the \emph{beginning} of episode $k$ (with  $\D_1 \triangleq \emptyset$). Recall that the environment $\E$ is random. In the following, let $\ex_k(\cdot) \triangleq \ex(\cdot|\D_k)$ denote the expectation w.r.t. the posterior distribution $\E \sim \PP(\cdot|\D_k)$, and let  
$\I_k(X;Y) \triangleq  \DD_{\text{KL}}(\PP((X,Y)\in \cdot \ |\D_k) \Vert \PP(X \in \cdot \ |\D_k) \times  \PP(Y \in \cdot \ |\D_k) )$ 
be the mutual information conditioned on $\D_k$. Note that $\D_k$ is also a random variable, thus $\I_k(X;Y)$ itself is a random variable.} We also point out that $\I_k(X;Y)$ is different from the \emph{conditional mutual information} $\I(X;Y|\D_k)$. In fact, we have $\I(X;Y|\D_k) = \ex_{\D_k}( \I_k(X;Y))$, where $\ex_{\D_k}(\cdot)$ denotes the expectation over the randomness of $\D_k$.  

\paragraph{Information ratio} At the heart of the proposed algorithm is the notion of \emph{joint information ratio}:
\begin{align}
    \Gamma_k(\mu, \nu, \chi) \triangleq \frac{(\ex_k[V_{1,\mu^*(\E),\nu}^{\E}(s_1) - V_{1,\mu,\nu}^{\E}(s_1)] )^2}{\I_k^{\mu,\nu}(\chi; \T^{k}_{H+1})}, \label{eq:ir}
\end{align}
which is defined for each episode $k \in [K]$. Here, $\chi$ is called the \emph{learning target}, of which the most natural choice is the environment $\E$. The superscript $\mu,\nu$ in $\I_k^{\mu,\nu}(\chi; \T^{k}_{H+1})$ means that the trajectory $\T^{k}_{H+1}$ is obtained by executing the policy $(\mu,\nu)$. Note that the information ratio $\Gamma_k$ explicitly depends on $(\mu,\nu,\chi)$, and also  implicitly depends on the past trajectory $\D_k$ since both $\ex_k[\cdot]$ and $\I_k^{\mu,\nu}(\cdot;\cdot)$ are calculated based on the random environment $\E \sim \PP(\cdot|\D_k)$. The numerator of $\Gamma_k$ measures the (squared) expected difference between the values induced by policy $(\mu^*(\E),\nu)$ and by policy $(\mu,\nu)$, which can be understood as the sub-optimality of $\mu$ for a fixed $\nu$. The denominator measures the information about $\chi$ learned by the policy $(\mu,\nu)$ through the interacted trajectory $\T_{H+1}^k$.  Roughly speaking, the joint information ratio measures the \emph{cost} of learning a unit of information about the learning target $\chi$ for the policy $(\mu,\nu)$.

For a fixed max-player's policy $\mu$, we further define the notion of \emph{marginal information ratio} w.r.t. $\mu$ as 
\begin{align}
\Lambda_k^{\mu}(\nu, \chi) \triangleq \frac{(\ex_k[V_{1,\mu,\nu}^{\E}(s_1) - V_{1,\mu,\dagger}^{\E}(s_1)] )^2}{\I_k^{\mu,\nu}(\chi; \T^{k}_{H+1})}.
\end{align}
Here, the numerator measures the expected regret induced by  the min-player's policy $\nu$  w.r.t. the fixed max-player's policy $\mu$. Again, the marginal information ratio $\Lambda_k^{\mu}(\nu, \chi)$ implicitly depends on the  past trajectory~$\D_k$.

\paragraph{Algorithm descriptions} We now introduce our algorithm {\sc MAIDS}, in which the learning target $\chi$ is selected as the environment $\E$. At the beginning of episode $k$, the max-player first calculates the posterior distribution of the environment  $\E \sim \PP(\cdot|\D_k)$ based on the past trajectory $\D_k$ and the prior distribution $\rho$, and then  chooses her policy as 
\begin{align}
    \muids = \argmin_{\mu \in \Psa} \max_{\nu \in \Psb} \ \Gamma_k(\mu,\nu,\E), \label{eq:ids1} 
\end{align}
implying that the max-player aims to minimize the joint information ratio by considering the presence of a worst-case opponent. This approach makes intuitive sense because the max-player has to make decisions prior to knowing the min-player's policy, so she should be more conservative. 

The min-player then chooses her policy based on the knowledge of $\muids$ and the posterior distribution $\E \sim \PP(\cdot|\D_k)$:
\begin{align}
 &\nuids = \argmin_{\nu \in \Psb} \Lambda_k^{\muids}(\nu, \E), 
\end{align}
which minimizes the marginal information ratio w.r.t. the max-player's policy $\muids$.

\paragraph{Regret bounds} A Bayesian regret bound for {\sc MAIDS} is provided in Theorem~\ref{thm:maids} below. 

\begin{theorem} \label{thm:maids}
Suppose the max-player's policy is $\mui = \{\muids\}_{k\in[K]}$ and the min-player's policy is $\nui = \{\nuids\}_{k\in[K]}$, then for any prior distribution~$\rho$, the Bayesian regret of $\mui$ satisfies 
\begin{align}
\mathsf{BR}_K(\mui) \le 8 S^{3/2} AB H^2 \sqrt{  K \log(SKH)}. 
\end{align}
\end{theorem}

\begin{proof}[Proof of Theorem~\ref{thm:maids}]
See Appendix~\ref{sec:proof_thm1} for the detailed proof.
\end{proof}

\begin{remark}
{\em \mc{In a symmetric fashion, the min-player can also run a {\sc MAIDS} algorithm (tailored to  the min-player) to obtain a policy $\tilde{\nu}_{\text{IDS}}$ that has a bounded Bayesian regret $\mathsf{BR}_K(\tilde{\nu}_{\text{IDS}})$ (as per Theorem~\ref{thm:maids}), with the assistance of the max-player. The joint policy $(\mui,\tilde{\nu}_{\text{IDS}})$ of the two players, where $\mui$ is the one in Theorem~\ref{thm:maids}, thus has a bounded Bayesian duality gap and is close to the NE policy.     Such an asymmetric learning structure has also been adopted in~\cite{jin2021power, huang2021towards, xiong2022self}; however, adapting this trick for IDS requires different analysis techniques compared to prior works that are based on OFU or Thompson sampling.  }  }
\end{remark}

\begin{remark}{\em
Note that the regret bound is valid for all possible prior distribution $\rho$ of the environment $\E$, and the scaling $\tilde{O}(\text{poly}(S,A,B,H) \cdot \sqrt{K})$  is order-optimal w.r.t. the number of episodes $K$. Meanwhile, we note that the scalings of $(S,A,B,H)$ have not yet matched the information-theoretic lower bound~\cite{bai2020provable}. This may be partially attributed to the looseness in the analyses. As demonstrated in a recent study on Thompson sampling in RL~\cite{moradipari2023improved}, it is possible to reduce a factor of $\sqrt{SAB}$ in Theorem~\ref{thm:maids}, by adopting refined proof techniques (based on posterior consistency tools). To further reduce the regret, it might be necessary to substitute the learning target from the entire environment $\E$ to certain compressed environments (as discussed in Section~\ref{sec:compressed}), since describing the transition kernels of the entire environment is costly. }
\end{remark}

\section{The {\sc Reg-MAIDS} Algorithm} \label{sec:regularized}

Although the basic algorithm {\sc MAIDS} has been proven to be sample-efficient, it has a high computational complexity due to the requirement of optimizing over the policy spaces of two players ($\Psa$ and $\Psb$) and calculating the mutual information term. To mitigate the computational load of {\sc MAIDS}, we propose a more (computationally) efficient IDS algorithm---called {\sc Reg-MAIDS}---that has the same Bayesian regret bound as {\sc MAIDS} in zero-sum MGs while enjoying less computational complexity.  Again, the learning target $\chi$  is selected as the environment $\E$. This approach is inspired by the prior work~\cite{hao2022regret} on single-agent RL. 

\paragraph{Algorithm descriptions} At the beginning of episode $k$, the max-player first calculates the posterior distribution of the environment $\E \sim \PP(\cdot|\D_k)$, and then chooses the policy
\begin{align*}
\muri = \argmax_{\mu \in \Psa} \min_{\nu \in \Psb} \big\{ \ex_k[ V_{1,\mu,\nu}^{\E}(s_1)] \!+\! \lambda \I_k^{\mu,\nu}(\E; \T^k_{H+1}) \big\}, 
\end{align*}
where $\lambda > 0$ is a parameter that controls the relative weights of the expected value and the regularization term (i.e., the  amount of information learned by the policy $(\mu,\nu)$). Compared to the joint information ratio term of {\sc MAIDS} in~\eqref{eq:ids1}, a key difference is that $\muri$ does not require the calculation of the value $V_{1,\mu^*(\E),\nu}^{\E}(s_1)$. As we will see later in this section, this modification  leads to a significant reduction in the complexity of searching for the optimal solution. 

The min-player then chooses her policy based on the knowledge of $\muri$ and the posterior distribution $\E \sim \PP(\cdot|\D_k)$:
\begin{align*}
\nuri = \argmin_{\nu \in \Psb} \big\{\ex_k[V^{\E}_{1,\muri,\nu}(s_1)] \!-\! \tl \I_{k}^{\muri,\nu}(\E; \T^k_{H+1}) \big\}, 
\end{align*}
where $\tl > 0$ is a different parameter. Below, we first present a regret bound for {\sc Reg-MAIDS}, and then introduce computationally efficient implementation methods. 

\paragraph{Regret bounds} Theorem~\ref{thm:proof_thm2} below shows that {\sc Reg-MAIDS} achieves the same Bayesian regret bound as {\sc MAIDS} when setting $\lambda$ and $\tl$ appropriately. 

\begin{theorem}\label{thm:proof_thm2}
Suppose the max-player's policy is $\murii = \{\muri\}_{k\in[K]}$ and the min-player's policy is $\nurii = \{\nuri\}_{k\in[K]}$, with $\lambda = \tl = \sqrt{2KH^2/S\log(SKH)}$. For any prior~$\rho$, the Bayesian regret of $\murii$ satisfies 
\begin{align}
\mathsf{BR}_K(\murii) \le 8 S^{3/2} AB H^2 \sqrt{  K \log(SKH)}. 
\end{align}
\end{theorem}
\begin{proof}[Proof of Theorem~2]
See Appendix~\ref{sec:proof_thm2} for the detailed proof.
\end{proof}

\paragraph{Equivalent formulas of $\muri$ and $\nuri$}
Given the past trajectory $\D_k$, we define a \emph{mean environment} as $\me = (H,\s,\A,\B,\{P_h^{\me}\}_{h=1}^H, \{r_h^{\me}\}_{h=1}^H)$, where the transition kernel
$$P_h^{\me}(\cdot | s,a,b) \triangleq \ex_{\E \sim \PP(\cdot|\D_k)}[P_h^{\E}(\cdot | s,a,b)], \quad \forall (s,a,b,h) \in (\s,\A,\B,[H]) $$ and the reward function $$\trh(s,a,b) \triangleq r_h(s,a,b) + \lambda \ex_k\left[ \DD_{\text{KL}}(P_h^{\E}(\cdot|s,a,b) \Vert P_h^{\me}(\cdot|s,a,b)) \right], \quad \forall (s,a,b,h) \in (\s,\A,\B,[H]).$$ 
\mc{The definition of $P_h^{\me}(\cdot | s,a,b)$ is also called \emph{Bayesian model averaging}~\cite{hoeting1999bayesian} in statistics}.
For the max-player, Lemma~\ref{lemma:eq} below shows that the objective function $ \ex_k\left[ V_{1,\mu,\nu}^{\E}(s_1) \right] + \lambda \I_k^{\mu,\nu}(\E; \T^k_{H+1})$ in the formula of $\muri$ is equal to the value $V_{1,\mu,\nu}^{\me}(s_1)$ in the mean environment $\me'$.

\begin{lemma}\label{lemma:eq}
    For any policy $(\mu,\nu)$, we have $$\ex_k\left[ V_{1,\mu,\nu}^{\E}(s_1) \right] + \lambda \I_k^{\mu,\nu}(\E; \T^k_{H+1}) = V_{1,\mu,\nu}^{\me}(s_1).$$ 
\end{lemma}

\begin{proof}[Proof of Lemma~\ref{lemma:eq}]
See Appendix~\ref{sec:proof_eq}.     
\end{proof}
Given Lemma~\ref{lemma:eq}, one can rewrite the formula of $\muri$ as 
\begin{align}
\muri = \argmax_{\mu \in \Psa} \min_{\nu \in \Psb} V_{1,\mu,\nu}^{\me}(s_1),     \label{eq:hu}
\end{align}
which is equivalent to finding a max-player's NE policy $\mu^*$ in the  environment $\me$.
For the min-player, we define $\me' = (H,\s,\A,\B,\{P_h^{\me'}\}_{h=1}^H, \{\rhdd\}_{h=1}^H)$ as another mean environment, where $P_h^{\me'} = P_h^{\me}$ and  $\rhdd(s,a,b) \triangleq r_h(s,a,b) - \tl \ex_k\left[ \DD_{\text{KL}}(P_h^{\E}(\cdot|s,a,b) \Vert P_h^{\me}(\cdot|s,a,b)) \right].$   Similar to Lemma~\ref{lemma:eq}, one can show that for any $\mu$,
\begin{align*}
\ex_k[V^{\E}_{1,\mu,\nu}(s_1)] - \tl \I_{k}^{\mu,\nu}(\E; \T^k_{H+1}) = V_{1,\mu, \nu}^{\me'}(s_1).
\end{align*}
Thus, the formula of $\nuri$ can be rewritten as
\begin{align}
    \nuri = \argmin_{\nu \in \Psb} V_{1,\muri, \nu}^{\me'}(s_1),
\end{align}
which is equivalent to finding the best response to the max-player's policy $\muri$ in the environment $\me'$.

\paragraph{Implementing {\sc Reg-MAIDS} efficiently}
Following the prior work~\cite{hao2022regret} in RL settings, we assume the existence of a \emph{posterior sampling oracle}, which, upon receiving a call, outputs a sample of the environment $\E$ from the posterior distribution $\PP(\cdot|\D_k)$. Multiple calls to the oracle lead to multiple \emph{independent} samples. As discussed in~\cite{hao2022regret}, the posterior sampling oracle can be obtained, either exactly or approximately, using epistemic neural networks~\cite{osband2021epistemic}. 

Given the posterior sampling oracle, the max-player can accurately approximate the mean environment $\me$ with transition kernel $\{P_h^{\me}\}_{h=1}^H$ and reward function $\{\trh \}_{h=1}^H$ using Monte Carlo sampling. Next, based on the equivalent formula of $\muri$ in~\eqref{eq:hu}, the max-player can find a NE policy $\mu^*$ in the environment $\me'$ efficiently by using any methods such as value-based and policy-based methods. Similarly, the min-player can first approximate the mean environment $\me'$ using Monte Carlo sampling, and then finds the best response to the max-player's policy $\muri$ in the environment $\me'$ using any efficient single-agent RL algorithms.

\section{Learning Compressed Environments} \label{sec:compressed}

An appealing feature of IDS is that it provides freedom and flexibility to select the learning target $\chi$. In situations where the environment is overwhelmingly complex and surpasses agents' computation capacity, it can be beneficial for agents to focus only on essential parts of the environment, while disregarding less significant details that have less impact on decision-making (as illustrated in Example~\ref{example:compress} below). This is analogous to the concept of lossy compression in information theory.

\mc{
\begin{example} \label{example:compress}
    Suppose an MG (environment) $\E$ can be decomposed into the \emph{product} of two sub-MGs: the main sub-MG $\mathcal{M}_1$ and side sub-MG $\mathcal{M}_2$, each with state space $\bar{\s}$ and action spaces $\bar{\A}$ and $\bar{\B}$. An illustration is provided in Figure~\ref{fig:product}. Assume the reward of side sub-MG $\mathcal{M}_2$ at each step is upper-bounded by $\delta$.  Intuitively,  as $\delta \to 0$, the impact of $\mathcal{M}_2$ on the overall decision-making becomes increasingly insignificant. Consequently, agents can prioritize learning about $\mathcal{M}_1$ (while ignoring $\mathcal{M}_2$) in order to make good decisions.   
\end{example}
}

\begin{figure}[t]
\includegraphics[width=8cm]{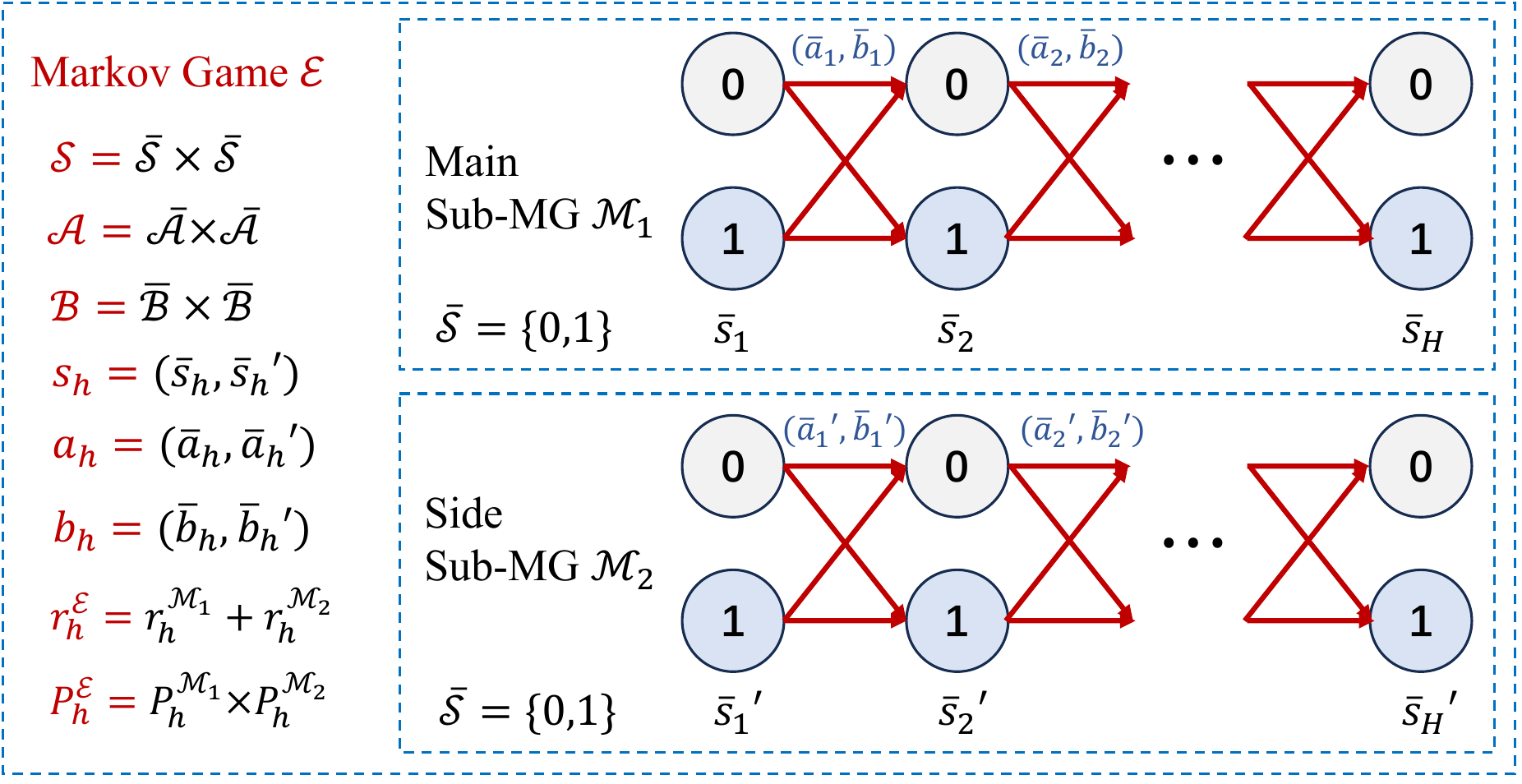}
\centering
\caption{Illustration of the MG in Example~\ref{example:compress}}
 \label{fig:product}
\end{figure}

Following this intuition, we introduce new IDS-based algorithms where the learning target $\chi$ is not the environment $\E$ itself but a compressed environment.\footnote{The \emph{surrogate environment} introduced in~\cite{hao2022regret,moradipari2023improved} is conceptually similar to the compressed environment. We call it compressed environment because it is analogous to lossy compression problem in information theory.} We first propose a distortion measure to evaluate the discrepancy between any two environments, and then put forth two principles for constructing compressed environments.

\paragraph{Distortion measure}
Inspired by  \emph{lossy compression} (a.k.a. rate-distortion theory) in information theory, we propose to compress the environment $\E$ to a compressed environment $\te$ that is ``close'' to $\E$ under a carefully defined distortion measure. 
Recall that $\Pi_A$ and $\Pi_B$ are the sets of max-player's and min-player's Markov policies respectively, and let $\Phi_A \subseteq \Pi_A$ and $\Phi_B \subseteq \Pi_B$ be any subsets of their policies.  For any environments $e, e' \in \Theta$, we define  
\begin{align*}
d_{\Phi_A, \Phi_B}(e,e') \triangleq \sup_{(\mu,\nu)\in \Phi_A \times \Phi_B} \big| V_{1,\mu,\nu}^e(s_1) - V_{1,\mu,\nu}^{e'}(s_1) \big|
\end{align*}
as a distortion measure  between $e$ and $e'$ through their value under the worst-case policy $(\mu,\nu) \in \Phi_A \times \Phi_B$. A comparison to the distortion measure proposed in~\cite{arumugam2022deciding} is provided in Remark~\ref{remark:distortion}, Appendix~\ref{sec:proof_thm3}.   

\begin{remark}{\em
A sufficient condition for the distortion $d_{\Phi_A, \Phi_B}(e,e')$ to be small is that the transition kernels $P_h^e(\cdot|s,a,b)$ and $P_h^{e'}(\cdot|s,a,b)$ are close enough for \emph{all} $(s,a,b,h)$.  However, this condition may not be necessary. \mc{For instance, consider two environments $e$ and $e'$ in the context of Example~\ref{example:compress}. Suppose $e$ and $e'$ share the same main sub-MG $\mathcal{M}_1$ but differ in their side sub-MGs. In this case, even though  $P_h^e(\cdot|s,a,b)$ and $P_h^{e'}(\cdot|s,a,b)$ are not close enough for some $(s,a,b,h)$, the distortion still satisfies  $d_{\Phi_A, \Phi_B}(e,e') \le  2H\delta$ for any $\Phi_A \subseteq \Pi_A$ and $\Phi_B \subseteq \Pi_B$.  }
}
\label{remark:distortion}
\end{remark}

\paragraph{Construction of compressed environments} Given the distribution of the environment $\E$ and the distortion measure $d_{\Phi_A, \Phi_B}$, we introduce two methods for constructing the compressed environment $\te$ (which is a random variable).

\begin{enumerate}
\item (Soft-compression) Following standard approaches in rate-distortion theory, we require the compressed environment $\te$ to be correlated with  $\E$ in such a way that 
\begin{align}
    \ex[d_{\Phi_A, \Phi_B}(\E,\te)] \le \epsilon \label{eq:constraint}
\end{align}
for some  tolerance level $\epsilon > 0$. \mc{In particular, a popular choice of $\te$ is the one that minimizes the mutual information with $\E$ subject to~\eqref{eq:constraint}, i.e.,}\footnote{In theory, one can efficiently compute the optimal $\te$ in~\eqref{eq:ves} using the Blahut-Arimoto algorithm~\cite{blahut1972computation,arimoto1972algorithm}. However, in practice, challenges arise when dealing with continuous random variables, and we refer the readers to~\cite{arumugam2022deciding} for a discussion of effective empirical methods. } 
\begin{align}
\mc{\te = \argmin_{\te': \ex[d_{\Phi_A, \Phi_B}(\E,\te')] \le \epsilon} \I(\E ;\te').} \label{eq:ves} 
\end{align}
The choice in~\eqref{eq:ves} is adopted by~\cite{arumugam2022deciding} for their value-equivalent sampling algorithm in RL, and is also the standard choice for many lossy compression problems in information theory~\cite{berger2003rate}.

\item (Hard-compression) Another way to define the compressed environment $\te$ is to require the distortion between $\E$ and  $\te$ to be small almost everywhere, i.e.,
\begin{align}
\PP(d_{\Phi_A, \Phi_B}(\E,\te) > \epsilon ) = 0. \label{eq:constraint2}
\end{align}
This is more stringent compared to~\eqref{eq:constraint}, since~\eqref{eq:constraint} only requires the distortion between the pair of random variables $(\E, \te)$ to be small in the \emph{average} sense.\footnote{This is also the reason why we name the two principles ``soft-compression'' and ``hard-compression'' respectively.} \mc{To achieve~\eqref{eq:constraint2}, a feasible approach is to partition the entire environment space~$\Theta$ into $\C$ disjoint subspaces $\{\Theta_c \}_{c=1}^{\C}$, such that any two environments $e$ and $e'$ in the same subspace satisfy $d_{\Phi_A, \Phi_B}(e,e') \le \epsilon$. Then, we arbitrarily select an element
$e_c \in \Theta_c$ as the reference of $\Theta_c$, and construct $\te$ such that} 
\begin{align}
 \mc{   \te = e_c \ \ \text{if and only if} \ \ \E \in \Theta_c. }  \label{eq:ves2}
\end{align}
A concrete construction of $\{\Theta_c \}_{c=1}^{\C}$ is provided in Appendix~\ref{appendix:F2} based on partitioning the transition kernels.
\end{enumerate}

\paragraph{The {\sc Compressed-MAIDS} Algorithm}  We now introduce the algorithm {\sc Compressed-MAIDS}.  \mc{At the beginning of episode $k$, the max-player first calculates the posterior distribution of the environment  $\E \sim \PP(\cdot|\D_k)$, and then chooses a compressed environment $\te$ via the above construction rules (e.g., following~\eqref{eq:ves} or~\eqref{eq:ves2}).} Next, the max-player chooses the policy 
\begin{align*}
&\muci = \argmin_{\mu \in \Psa} \max_{\nu \in \Psb} \tilde{\Gamma}_k(\mu,\nu,\te), \\
&\text{where } \tilde{\Gamma}_k(\mu,\nu,\te) \triangleq \frac{(\ex_k[V_{1,\mu^*(\E),\nu}^{\te}(s_1) - V_{1,\mu,\nu}^{\te}(s_1)])^2}{\I_k^{\mu,\nu}(\te; \T^{k}_{H+1})}.
\end{align*}
The min-player calculates the distributions of $\E$ and $\te$ in the same way as the max-player, and chooses her policy as:
\begin{align*}
&\nuci = \argmin_{\nu \in \Psb} \tilde{\Lambda}_k^{\muci} (\nu, \te), \\
 &\text{where }  \tilde{\Lambda}_k^{\mu} (\nu, \te) \triangleq
 \frac{(\ex_k[V_{1,\mu,\nu}^{\te}(s_1) - V_{1,\mu,\nu^{\dagger}_{\E}(\mu) }^{\te}(s_1)] )^2}{\I_k^{\mu,\nu}(\te; \T^{k}_{H+1})}.
\end{align*}
Here, $\nu^{\dagger}_{\E}(\mu)$ denotes the best response to the max-player's policy $\mu$ in the environment $\E$.

\paragraph{Regret bounds} Next, we present Bayesian regret bounds for {\sc Compressed-MAIDS} \mc{when the compressed environment $\te$ is constructed through hard-compression as in~\eqref{eq:ves2}, while regret bounds in the context of soft-compression will be considered in future works.    Below, we express $\te$ as $\te_{\epsilon}$ to clearly indicate its correlation with the parameter~$\epsilon$.}


\begin{theorem}\label{thm:compressed}
Let the max-player's policy is $\mucii = \{\muci\}_{k=1}^K$ and the min-player's policy is $\nucii = \{\nuci\}_{k=1}^K$, where the compressed environment $\te_{\epsilon}$ at each episode $k\in[K]$ is constructed according to~\eqref{eq:ves2} with $\epsilon>0$. Then for any prior $\rho$, the Bayesian regret of $\mucii$ satisfies
\begin{align}
\mathsf{BR}_K(\mucii) &\le  4 \sqrt{KH^3 SAB \cdot \I(\te_{\epsilon}; \D_{K+1})} + 4K\epsilon, \notag 
\end{align}
where $\I(\te_{\epsilon}; \D_{K+1})$ is the mutual information w.r.t. executing the policy $(\mucii,\nucii)$.
\end{theorem}
The proof of Theorem~\ref{thm:compressed} is provided in Appendix~\ref{sec:proof_thm3}. 
Moreover, since one has the freedom to optimize $\epsilon$, we have
\begin{align*}
\mathsf{BR}_K(\mucii) &\le \inf_{\epsilon \ge 0} \Big\{ 4 \sqrt{KH^3 SAB \cdot \I(\te_{\epsilon}; \D_{K+1})} \!+\! 4K\epsilon \Big\} \\
&\le 4 \sqrt{KH^3 SAB \cdot \I(\te_{1/K}; \D_{K+1})} + 4.
\end{align*}

\paragraph{Advantages of {\sc Compressed-MAIDS}} \mc{We first discuss the advantages of {\sc Compressed-MAIDS} in the context of our illustrative example,  Example~\ref{example:compress}, when the parameter $\delta = 1/2HK$.   Considering the special structure of the MG in this example, an intuitive approach to compress the environment $\E$ is to retain the main sub-MG $\mathcal{M}_1$ while discard the (negligible) side sub-MG $\mathcal{M}_2$. Following the hard-compression principle, we first construct subspaces such that each subspace comprises all environments that have a same main sub-MG $\mathcal{M}_1$, upon which we construct a compressed environment $\te$ based on~\eqref{eq:ves2}.\footnote{For the purpose of illustration, we allow the number of subspaces to be infinite here (as the number of $\mathcal{M}_1$ may be infinite).} One can check that any pair of environments $e$ and $e'$ in the same subspace satisfy $d_{\Phi_A, \Phi_B}(e,e') \le 2H\delta = 1/K$, thus $\te$ can be expressed as $\te_{1/K}$. Applying an upper bound on the mutual information (see Lemma~\ref{lemma:mi}, Appendix~\ref{sec:proof_thm1}), we have $\I(\te_{1/K}; \D_{K+1}) \le 2|\bar{\s}|^2 |\bar{\A}| |\bar{\B}| H \log(|\bar{\s}| KH)$ where $|\bar{\s}| = \sqrt{S}$, $|\bar{\A}| = \sqrt{A}$ and $|\bar{\B}| = \sqrt{B}$. Thus, the Bayesian regret of $\mucii$ becomes $\tilde{O}(\sqrt{H^4S^2(AB)^{3/2}K})$, which reduces a factor of $\sqrt{S(AB)^{1/2}}$ compared to the regret of the basic algorithm {\sc MAIDS} (Theorem~\ref{thm:maids}). This suggests {\sc Compressed-MAIDS} is more efficient, as employing {\sc MAIDS} would lead the agents to concurrently learn about both the main and side sub-MGs, an approach that is less efficient considering the negligibility of the side sub-MG.}

It is worth noting that the construction of $\te$ in~\eqref{eq:ves2} is just one specific choice of the compressed environment; consequently, its regret may not be optimal. However, the IDS principle offers the flexibility in selecting methods for constructing the compressed environment, and different constructions may lead to different performances/regrets. In the single-agent RL setting,~\cite[Sec~4.1]{hao2022regret} introduces a novel approach that helps to reduce the Bayesian regret by a factor of $\sqrt{S}$; however, their approach cannot be easily applied to MGs due to the presence of unpredictable opponents. While  MARL is more complicated, we still conjecture that it is possible to construct alternative compressed environments satisfying~\eqref{eq:constraint} or~\eqref{eq:constraint2} that can lead to better regret bounds.  This opens up an intriguing avenue for future research.

\section{General-Sum Markov Games} \label{sec:general}

In this section, we introduce the mathematical formulation of the multi-player general-sum MG, and then show that an extended version of {\sc Reg-MAIDS} (Section~\ref{sec:regularized}) is provably sample-efficient in this general setting. In the following, let $\N \ge 2$ be the number of players.

\paragraph{Model} The general-sum MG can be represented by a tuple $\E = (\N,H,\s,\A, \{P_h\}_{h\in [H]}, \{r_h^{(i)}\}_{h\in[H],i\in[\N]} )$, where $\A = \otimes_{i=1}^N \A_i$ is the joint action spaces of the $\N$ players,  $P_h: \s \times \A  \to \Delta(\s)$ is the transition kernel, and $r_h^{(i)}: \s \times \A \to [0,1]$ is the deterministic reward function for player $i$. Let $S \triangleq |\s|$ and $A \triangleq|\A|$.  Similar to the zero-sum MG setting (Section~\ref{sec:model}), we assume $\{r_h^{(i)}\}_{h\in[H],i\in[\N]}$ are known, and we consider a Bayesian setting where we have a prior distribution $\rho$ on the environment $\E$ (i.e., on the transition kernels $\{P_h\}_{h\in [H]}$).

For player $i$, let $\Pipi \subseteq \{\s \to \Delta(\A_i) \}^H$ be a set of (Markov) \emph{pure policies}. She is also allowed to take a \emph{random policy} by adopting a random seed $\omega \in \Omega$. Specifically, a random policy $\pi^{(i)} = \{\pi_h^{(i)} \}_{h\in[H]}$ can be viewed as a mapping from $\Omega$ to $\Pipi$, i.e., $\pi^{(i)}(\omega) \in \Pipi$. When executing a random policy $\pi^{(i)}$, she first samples a random seed $\omega$ and then follows $\pi^{(i)}(\omega) = \{\pi_h^{(i)}(\omega) \}_{h\in[H]}$ in the $H$ steps. Next, we consider the $\N$ players jointly, and introduce the concepts of \emph{joint policy} and \emph{product policy}. The joint policy $\pi = \{ \pi^{(i)}(\omega_i) \}_{i\in[\N]}$ consists of $\N$ random policies where random variables $(\omega_1, \ldots, \omega_{\N})$ are correlated, while the product policy $\pi = \{ \pi^{(i)}(\omega_i) \}_{i\in[\N]}$ consists of $\N$ random policies where $(\omega_1, \ldots, \omega_{\N})$ are independent. As such, product policies are special cases of joint policies.

For policy $\pi$ and player $i$, the value function is denoted as 
\begin{align*}
    V^{(i),\E}_{h,\pi}(s) = \ex_\pi^{\E} \left[ \sum_{h'=h}^H r_{h'}^{(i)}(s_{h'}, a_{h'}) \big| s_h =s \right], \ \forall s \in \s.
\end{align*}
For any policy $\pi$, let $\pi^{(-i)}$ be the joint policy excluding player~$i$, and let $$\pi^{(i),\dagger}_{\E} \triangleq \argmax_{\nu \in \Delta(\Pipi)} V^{(i),\E}_{1, \nu \times \pi^{(-i)}}(s_1)$$ be the best response of player $i$ when facing $\pi^{(-i)}$, where $s_1$ is the fixed initial state. 

\paragraph{Learning objectives}
For the general-sum MG, we first define the concept of (approximate) NE and coarse correlated equilibrium  (CCE).
\begin{definition}[$\varepsilon$-NE and $\varepsilon$-CCE]
A policy $\pi$ is an $\varepsilon$-NE of the environment $\E$ if it is a \emph{product policy} and  it satisfies 
\begin{align}
V^{(i),\E}_{1,(\pi^{(i),\dagger}_{\E}, \pi^{(-i)})}(s_1) - V^{(i),\E}_{1,\pi}(s_1) \le \varepsilon, \ \forall i \in [\N]. \label{eq:ne}
\end{align}
We say a policy $\pi$ is an $\varepsilon$-CCE if it is a \emph{joint policy} and it satisfies~\eqref{eq:ne}.
\end{definition}

Suppose the $\N$ players interact with the environment for $K$ episodes. For a \emph{product policy} $\pi = \{\pi^k\}_{k\in[K]}$ and a realization of the environment $\E = e$, we define 
\begin{align}
\reg^{\text{NE}}_K(\e, \pi) \triangleq \sum_{k=1}^K \sum_{i=1}^{\N} V^{(i),e}_{1, (\pi_{e}^{k,(i),\dagger}, \pi^{k,(-i)})}(s_1) - V^{(i),e}_{1,\pi^k}(s_1) \notag 
\end{align}
as the cumulative regret in environment $e$, and define 
\begin{align}
\mathsf{BR}^{\text{NE}}_K (\pi) \triangleq  \ex_{\E \sim \rho} \left(\reg_K^{\text{NE}}(\E,\pi) \right) \label{eq:nash}
\end{align}
as the Bayesian regret, which is averaged over the prior distribution $\rho$ of the environment. For a \emph{joint policy} $\pi$, the definitions of regret $\reg^{\text{CCE}}_K(\e, \pi)$ and Bayesian regret $\mathsf{BR}^{\text{CCE}}_K (\pi)$ are exactly the same as $\reg^{\text{NE}}_K(\e, \pi)$ and $\mathsf{BR}^{\text{NE}}_K(\pi)$ (the only difference is that $\pi$ is a joint policy). 

\begin{table}[t!]
\centering
\caption{Notations for general-sum MGs } \vspace{3pt}
\begin{tabular}{ p{3cm}  p{12cm}  }
 \hline\hline
 $\N$   & Number of players \\
 $\s$   & State space \\
 $\A_i$ & Action space of player $i$ \\
 $\A$&   Joint action spaces of the $\N$ players  \\
 $P_h$&   Transition kernel at step $h$ ($P_h: \s \times \A \to \Delta(\s)$) \\
 $r_h^{(i)}$&   Reward function for player $i$ at step $h$ ($r_h^{(i)}: \s \times \A \to [0,1]$) \\
 $\Pipi$ & Set of (Markov) \emph{pure policies} of player $i$ \\
 $\omega$ & Random seed \\
  $V_{h,\pi}^{(i),\E}(s)$ & Value of state $s$ for player $i$ (at step $h$, w.r.t. policy $\pi$ and environment $\E$) \\
  $\pi^{(-i)}$ & Policy $\pi$ excluding
player $i$ \\
$\pi^{(i),\dagger}_{\E}$ & Best response of player $i$ when facing $\pi^{(-i)}$ in environment $\E$ \\

 \hline\hline
\end{tabular}
\label{table:1}
\end{table}

\paragraph{Algorithm descriptions} We now extend {\sc Reg-MAIDS} to the general-sum MG. At the beginning of episode $k$,  the posterior distribution of the environment $\E \sim \PP(\cdot|\D_k)$ can be calculated based on  $\D_k$ and prior distribution~$\rho$. Similar to Section~\ref{sec:regularized}, we construct a mean environment $\meh$ such that the transition kernel and reward satisfy 
\begin{align}
&P_h^{\meh}(\cdot |s,a) \triangleq \ex_{\E \sim \PP(\cdot|\D_k)} [P_h^{\E}(\cdot|s,a)], \\
&r_h^{(i),\meh}(s,a) \triangleq r_h^{(i)}(s,a) + \lambda \ex_k[ \DD_{\text{KL}}(P_h^{\E}(\cdot|s,a) \Vert P_h^{\meh}(\cdot|s,a)) ], \ \forall i \in [\N].
\end{align}
The algorithm proceeds in two steps: 
\begin{enumerate}
    \item For each pure policy $\pi \in \otimes_{i=1}^{\N} \Pipi$, we calculate the  value $V_{1,\pi}^{(i),\meh}(s_1)$ in the mean environment $\meh$ for all the players;

    \item We view the MG on $\meh$ as a \emph{normal-form game} with the space of pure strategies being $\{\Pipi\}_{i\in[\N]}$ and payoff functions being $\{V_{1,\pi}^{(i),\meh}(s_1) \}_{i\in[\N],\pi\in \otimes_{i=1}^N \Pipi}$. We then compute a NE strategy $\pi^{\text{NE}}$ (resp. a CCE strategy $\pi^{\text{CCE}}$) for this normal-form game, and let the output policy for episode~$k$ in the MG be $\pidskne = \pi^{\text{NE}}$ (resp. $\pidskcce =\pi^{\text{CCE}}$).
\end{enumerate}

Theorem~\ref{thm:general} below provides Bayesian regret bounds for $\pids^{\text{NE}} = \{\pidskne\}_{k\in[K]}$ and $\pids^{\text{CCE}} = \{\pidskcce\}_{k\in[K]}$.

\begin{theorem} \label{thm:general}
For any prior distribution~$\rho$, the Bayesian regrets of $\pids^{\text{NE}}$ and $\pids^{\text{CCE}}$  respectively satisfy
\begin{align}
&\mathsf{BR}^{\text{NE}}_K(\pids^{\text{NE}}) \le  3\N S^{3/2} A H^2 \sqrt{  K \log(SKH)}, \quad \text{and} \\
&\mathsf{BR}^{\text{CCE}}_K(\pids^{\text{CCE}}) \le 3\N S^{3/2} A H^2 \sqrt{  K \log(SKH)}
\end{align}
by choosing $\lambda = \sqrt{HK^2/S\log(SKH)}$.
\end{theorem}
The proof of Theorem~\ref{thm:general} is provided in Appendix~\ref{appendix:general}. 

Similar to the regret bound of {\sc Reg-MAIDS} for the two-player zero-sum MG, the regret bound in Theorem~\ref{thm:general} for the multi-player general-sum MG is also valid for all possible prior distribution $\rho$ of the environment $\E$, and the scaling $\tilde{O}(\sqrt{K})$  is also order-optimal w.r.t. the number of episodes $K$.  Additionally, we point out that while this paper focuses on learning NE/CCE, it is also possible to learn a \emph{correlated equilibrium} (CE) policy with theoretical guarantees via similar techniques.

\section{Conclusion and Future Works} \label{sec:conclusion}

This work presents sample-efficient MARL algorithms based on the design principle of information-directed sampling (IDS). These algorithms utilize information-theoretic concepts to effectively navigate the exploration-exploitation tradeoff by balancing the acquired information about the environment (exploration) and the policy’s sub-optimality (exploitation) when choosing policies. Theoretically, we prove that these IDS-based algorithms achieve order-optimal Bayesian regret with respect to the number of episodes in both two-player zero-sum MGs and multi-player general-sum MGs. This contribution enriches  the set of sample-efficient algorithms available for tackling MARL problems. Moreover, we also provide a computationally efficient variant, {\sc Reg-MAIDS}, that is not only theoretically sound but can also be implemented practically.  

Finally, we put forth two promising directions for future research.
\begin{enumerate}
    \item While the Bayesian regret bounds for our IDS-based algorithms are order-optimal with respect to the number of episodes $K$, they do not achieve order-optimality in terms of the cardinality of state space $S$, the cardinality of action spaces, and the episode length $H$. One would expect to close this gap by either refining the algorithm design (perhaps choosing other compressed environments as the learning target in {\sc Compressed-MAIDS}) or developing  refined analytical techniques (e.g., adapting recent analyses for Thompson sampling to our IDS-based algorithms). Working in this direction would be a fruitful endeavour for better justifying the theoretical advantages of the IDS principle. 

\item As an initial investigation of the IDS principle in MARL, this work mainly focuses on the \emph{tabular} setting where the state space and the action space are discrete and finite. To address more practical scenarios where the state space is exceedingly large, it would be beneficial to develop new IDS-based algorithms for MG with \emph{linear function approximation} or \emph{general function approximation}, enhancing their applicability in real-world scenarios.
\end{enumerate}

\section*{Acknowledgement}
The authors would like to thank Dr.~Zhuoran Yang for his valuable discussions and insights, which have significantly improved the quality of this paper.





\vspace{18pt}
\noindent{\bf \LARGE	 Appendices}
\appendix

\section{Appendix for {\sc MAIDS}: Proof of Theorem~\ref{thm:maids}} \label{sec:proof_thm1}

The proof  for Theorem~\ref{thm:maids} is provided in the following.

First note that for any policy $(\mu,\nu)=(\{\mu^k\}_{k\in[K]}, \{\nu^k\}_{k\in[K]})$,  the Bayesian regret can be decomposed as follows:
\begin{align}
\mathsf{BR}_K (\mu) &= \ex \left(\sum_{k=1}^K V_1^{\E,*}(s_1) - V_{1,\mu^k, \dagger}^{\E} (s_1) \right) \\
&= \ex \left(\sum_{k=1}^K V_1^{\E,*}(s_1) - V_{1,\mu^k, \nu^k}^{\E} (s_1) \right) + \ex \left(\sum_{k=1}^K V_{1,\mu^k, \nu^k}^{\E} (s_1) - V_{1,\mu^k, \dagger}^{\E} (s_1) \right). \label{eq:ji}
\end{align}
The first term in~\eqref{eq:ji} represents the expected difference between the Nash value $V_1^{\E,*}(s_1)$ and the value induced by the policy $(\mu^k,\nu^k)$. The second term represents the Bayesian regret induced by the min-player's policy $\nu^k$ with respect to the fixed max-player's policy $\mu^k$.

\subsection{The first term in~\eqref{eq:ji}} \label{sec:4.1}
For the first term, we have 
\begin{align}
&\ex \left(\sum_{k=1}^K V_1^{\E,*}(s_1) - V_{1,\mu^k, \nu^k}^{\E} (s_1) \right) \label{eq:18} \\
&\le \ex \left(\sum_{k=1}^K V_{1,\mu^*(\E),\nu^k}^{\E}(s_1) - V_{1,\mu^k, \nu^k}^{\E} (s_1) \right) \label{eq:new} \\
&= \sum_{k=1}^K \ex_{\D_k} \left[ \ex_{\E \sim \PP(\cdot|\D_k )} \left(V_{1,\mu^*(\E),\nu^k}^{\E}(s_1) - V_{1,\mu^k, \nu^k}^{\E} (s_1) \right) \right] \\
&\le \sum_{k=1}^K \ex_{\D_k} \left[\sqrt{\frac{\left(\ex_{k} \left(V_{1,\mu^*(\E),\nu^k}^{\E}(s_1) - V_{1,\mu^k, \nu^k}^{\E} (s_1) \right)\right)^2}{\I_k^{\mu^k,\nu^k}(\E;\T^{k}_{H+1})}} \sqrt{\I_k^{\mu^k,\nu^k}(\E;\T^{k}_{H+1})} \right] \\
&\le \sqrt{\sum_{k=1}^K \ex_{\D_k}\left[\frac{\left(\ex_{k} \left(V_{1,\mu^*(\E),\nu^k}^{\E}(s_1) - V_{1,\mu^k, \nu^k}^{\E} (s_1) \right)\right)^2}{\I_k^{\mu^k,\nu^k}(\E;\T^{k}_{H+1})}\right] } \cdot \sqrt{\sum_{k=1}^K \ex_{\D_k} \left[\I_k^{\mu^k,\nu^k}(\E;\T^{k}_{H+1}) \right]} \label{eq:cauchy} \\
&= \sqrt{\sum_{k=1}^K \ex_{\D_k} \left[\Gamma_k(\mu^k,\nu^k,\E) \right] } \cdot \sqrt{\sum_{k=1}^K \ex_{\D_k} \left[\I_k^{\mu^k,\nu^k}(\E;\T^{k}_{H+1}) \right]}, \label{eq:ids}
\end{align}
where the expectation $\ex_k(\cdot)$ is over the randomness of the environment $\E$ conditioned on the past trajectory $\D_k$, thus  $\ex_k(\cdot)$ is equivalent to $\ex_{\E \sim \PP(\cdot|\D_k)}(\cdot)$ and these two notations may be used interchangeably. Eqn.~\eqref{eq:new} follows from the property of the NE policy introduced in~\eqref{eq:6}. Eqn.~\eqref{eq:cauchy} follows from the Cauchy-Schwarz inequality, and Eqn.~\eqref{eq:ids} is due to the definition of information ratio. 

For the purpose of analyses, we introduce another proxy policy for the max-player---the Thompson sampling (TS) policy $\muts$. The TS policy $\muts$ first samples a realization of the environment $\E = \e$ according to the distribution $\E \sim \PP(\cdot|\D_k)$, and then chooses the max-player's NE policy $\mu^*(e)$ with respect to $\e$. Note that the TS policy $\muts$ is a Markov policy.\footnote{Strictly speaking, the TS policy $\muts$ is a Markov policy that incorporates an additional random seed for sampling an environment $\E = e$. For rigorousness, one needs to slightly modify the definition of the set of Markov policies $\Pi_A$  to include all Markov policies with random seeds, thus encompassing the TS policy $\muts \in \Pi_A$. However, for ease of presentation, we maintain the definition of $\Pi_A$ as presented in Section~\ref{sec:model} of the main text, and such imprecision can be easily fixed. }

Since the max-player chooses $\muids$ and the min-player chooses $\nuids$ at episode $k$, we have
\begin{align}
 \sqrt{\sum_{k=1}^K \ex_{\D_k} \left[\Gamma_k(\muids,\nuids,\E) \right] } &\le \sqrt{\sum_{k=1}^K \ex_{\D_k} \left[\max_{\nu \in \Psb} \Gamma_k(\muids,\nu,\E) \right] } \\
&\le  \sqrt{\sum_{k=1}^K \ex_{\D_k} \left[\max_{\nu \in \Psb} \Gamma_k(\muts,\nu,\E) \right] }, \label{eq:jian}
\end{align}
where Eqn.~\eqref{eq:jian} is due to the fact that $\muids$ minimizes $\max_{\nu}\Gamma_k(\mu,\nu,\E)$ by definition. Below, we show that the information ratio of the Thompson sampling policy $\muts$ is bounded for any $\nu \in \Psb$ and any distribution~$\E$.

\begin{lemma} \label{lemma:ts}
For any $\nu \in \Psb$ and any distribution $\E$, we have $$\Gamma_k(\muts,\nu,\E) \le 4H^3SAB.$$
\end{lemma}
\begin{proof}[Proof of Lemma~\ref{lemma:ts}]
The detailed proof is deferred to Appendix~\ref{sec:proof_ts}.
\end{proof}

Applying Lemma~\ref{lemma:ts} to Eqn.~\eqref{eq:jian}, we obtain
\begin{align}
\sqrt{\sum_{k=1}^K \ex_{\D_k} \left[\Gamma_k(\muids,\nuids,\E) \right] } \le \sqrt{4KH^3SAB}.    \label{eq:dk}
\end{align}

For the second term of Eqn.~\eqref{eq:ids}, by standard information inequalities, we have
\begin{align}
\sum_{k=1}^K \ex_{\D_k}  \left[\I_k^{\mu^k,\nu^k}(\E;\T^{k}_{H+1}) \right]  &= \sum_{k=1}^K \I^{\mu^k,\nu^k}(\E;\T^{k}_{H+1}|\D_k)  \label{eq:mi1}\\
&= \sum_{k=1}^K \I^{\mu^k,\nu^k}(\E;\T^{k}_{H+1}|\T^{1}_{H+1},\ldots, \T^{k-1}_{H+1}) \\
&= \I^{\mu,\nu}(\E;\T^{1}_{H+1},\ldots, \T^{k}_{H+1}) \label{eq:mi2} \\
&=
\I^{\mu,\nu}(\E;\D_{K+1}),  \label{eq:mi3}  
\end{align}
where~\eqref{eq:mi1} follows from the definition of conditional mutual information, and~\eqref{eq:mi2} follows from the chain rule of mutual information.

\begin{lemma} \label{lemma:mi}
For any policy $(\mu,\nu)$, the mutual information between the environment $\E$ and the whole trajectory $\D_{K+1}$ induced by $(\mu,\nu)$ is upper-bounded as 
\begin{align}
\I^{\mu,\nu}(\E;\D_{K+1}) \le 2S^2ABH\log(SKH).    
\end{align}
\end{lemma}
\begin{proof}[Proof of Lemma~\ref{lemma:mi}]
The proof is adapted from~\cite[Lemma~3.3]{hao2022regret} with appropriate modifications for two-player zero-sum MGs.    
\end{proof} 

Note that Lemma~\ref{lemma:mi} holds for any policy $(\mu,\nu)$, thus it also applies to the executed policy $(\mui,\nui) = (\{\muids \}_{k\in[K]}, \{\nuids\}_{k\in[K]})$. 

Therefore, combining Eqns.~\eqref{eq:18}-\eqref{eq:ids},~\eqref{eq:dk},~\eqref{eq:mi1}-\eqref{eq:mi2}, and Lemma~\ref{lemma:mi}, as well as substituting $(\mu,\nu)$ to $(\mui,\nui)$, we have 
\begin{align}
\ex \left(\sum_{k=1}^K V_1^{\E,*}(s_1) - V_{1,\muids, \nuids}^{\E} (s_1) \right) &\le \sqrt{4KH^3SAB} \cdot \sqrt{2S^2ABH\log(SKH)} \\
&\le 4S^{3/2}ABH^2 \sqrt{K \log(SKH) }.
\end{align}

\subsection{The second term in~\eqref{eq:ji}}
The second term in~\eqref{eq:ji}, which represents the Bayesian regret induced by the min-player's policy $\nuids$ with respect to the fixed max-player's policy $\muids$, can be handled in a similar manner as in Section~\ref{sec:4.1}. Below, we provide a brief sketch of the proof. 

Following similar steps in Eqns.~\eqref{eq:18}-\eqref{eq:ids} and Eqns.~\eqref{eq:mi1}-\eqref{eq:mi3}, we have
\begin{align}
&\ex\left(\sum_{k=1}^K V_{1,\muids, \nuids}^{\E} (s_1) - V_{1,\muids, \dagger}^{\E} (s_1) \right) \\
&= \sum_{k=1}^K \ex_{\D_k} \left[ \ex_{\E \sim \PP(\cdot|\D_k )} \left(V_{1,\muids,\nuids}^{\E}(s_1) - V_{1,\muids, \dagger}^{\E} (s_1) \right) \right] \\
&= \sum_{k=1}^K \ex_{\D_k} \left[\sqrt{\frac{\left(\ex_{k} \left(V_{1,\muids,\nuids}^{\E}(s_1) - V_{1,\muids, \dagger}^{\E} (s_1) \right)\right)^2}{\I_k^{\muids,\nuids}(\E;\T^{k}_{H+1})}} \sqrt{\I_k^{\muids,\nuids}(\E;\T^{k}_{H+1})} \right] \\
&\le \sqrt{\sum_{k=1}^K \ex_{\D_k}\left[\frac{\left(\ex_{k} \left(V_{1,\muids,\nuids}^{\E}(s_1) - V_{1,\muids, \dagger}^{\E} (s_1) \right)\right)^2}{\I_k^{\muids,\nuids}(\E;\T^{k}_{H+1})}\right] } \cdot \sqrt{\sum_{k=1}^K \ex_{\D_k} \left[\I_k^{\muids,\nuids}(\E;\T^{k}_{H+1}) \right]}\\
&\le \sqrt{\sum_{k=1}^K \ex_{\D_k} \left[\Lambda^{\muids}_k(\nuids,\E) \right] } \cdot \sqrt{\I^{\mui,\nui}(\E;\D^{K+1}) }.    
\end{align}
Recall that for each episode $k$ and for every possible trajectory $\D_k$, the max-player's policy $\muids$ can be calculated based on the minimax optimization in~\eqref{eq:ids1}. Let $\nuts(\muids)$ be the TS policy of the min-player with respect to the max-player's policy $\muids$, such that it first samples a realization of the environment $\E = \e$ according to $\E \sim \PP(\cdot|\D_k)$, and then chooses the best response to $\muids$ under environment $\e$ (which is denoted by $\nu^{\dagger}_e(\muids)$).  When there is no confusion, we may use the abbreviation $\nuts$ to replace $\nuts(\muids)$ for simplicity. Note that the TS policy $\nuts$ is a Markov policy. By the definition of $\nuids$, we have
\begin{align}
\Lambda^{\muids}_k(\nuids,\E) \le \Lambda^{\muids}_k(\nuts(\muids),\E).    
\end{align}
Moreover, one can prove an analogous version of Lemma~\ref{lemma:ts} to show that the TS policy $\nuts$ satisfies $$\Lambda^{\muids}_k(\nuts(\muids),\E) \le 4H^3SAB.$$ Combining the upper bound on $\Lambda^{\muids}_k(\nuids,\E)$ with Lemma~\ref{lemma:mi}, one can eventually show that   
\begin{align}
    \ex\left(\sum_{k=1}^K V_{1,\muids, \nuids}^{\E} (s_1) - V_{1,\muids, \dagger}^{\E} (s_1) \right) \le 4S^{3/2}ABH^2 \sqrt{K \log(SKH) }.
\end{align}


\section{Appendix for {\sc Reg-MAIDS}: Proof of Theorem~\ref{thm:proof_thm2}} \label{sec:proof_thm2}

We first focus on the first term in~\eqref{eq:ji}. Note that for any min-player's policy $\nu = \{\nu^k\}_{k=1}^K$, the max-player's policy $\murii = \{\muri\}_{k=1}^K$ satisfies
\begin{align}
&\ex \left(\sum_{k=1}^K V_1^{\E,*}(s_1) - V_{1,\muri, \nu^k}^{\E} (s_1) \right) \notag \\
&= \sum_{k=1}^K \ex_{\D_k} \left[ \ex_{k} \left(V_1^{\E,*}(s_1) - V_{1,\muri, \nu^k}^{\E} (s_1) \right) \right] \notag \\
&= \sum_{k=1}^K \ex_{\D_k}\! \left[\ex_{k}\! \left[V_1^{\E,*}(s_1)\right] -\ex_k\!\left[ V_{1,\muri, \nu^k}^{\E} (s_1) \right] - \lambda \I^{\muri,\nu^k}_k(\E; \T^k_{H+1})  + \lambda \I^{\muri,\nu^k}_k(\E; \T^k_{H+1}) \right] \notag \\
&\le \sum_{k=1}^K \ex_{\D_k}\! \left[\ex_{k}\! \left[V_1^{\E,*}(s_1)\right] - \left( \min_{\nu'} \ex_k\!\left[ V_{1,\muri, \nu'}^{\E} (s_1) \right] + \lambda \I^{\muri,\nu'}_k(\E; \T^k_{H+1}) \right)  + \lambda \I^{\muri,\nu^k}_k(\E; \T^k_{H+1}) \right]. \label{eq:mine} 
\end{align}
By the definition of $\muri$, we have 
$$\min_{\nu'} \ex_k\!\left[ V_{1,\muri, \nu'}^{\E} (s_1) \right] + \lambda \I^{\muri,\nu'}_k(\E; \T^k_{H+1}) \ge \min_{\nu'} \ex_k\!\left[ V_{1,\muids, \nu'}^{\E} (s_1) \right] + \lambda \I^{\muids,\nu'}_k(\E; \T^k_{H+1}).$$
Thus,~\eqref{eq:mine} can be bounded from above as
\begin{align}
&\sum_{k=1}^K \ex_{\D_k}\! \left[\ex_{k}\! \left[V_1^{\E,*}(s_1)\right] - \left( \min_{\nu'} \ex_k\!\left[ V_{1,\muri, \nu'}^{\E} (s_1) \right] + \lambda \I^{\muri,\nu'}_k(\E; \T^k_{H+1}) \right)  + \lambda \I^{\muri,\nu^k}_k(\E; \T^k_{H+1}) \right] \notag \\
&\le  \sum_{k=1}^K \ex_{\D_k}\! \left[\ex_{k}\! \left[V_1^{\E,*}(s_1)\right] - \left( \min_{\nu'} \ex_k\!\left[ V_{1,\muids, \nu'}^{\E} (s_1) \right] + \lambda \I^{\muids,\nu'}_k(\E; \T^k_{H+1}) \right)  + \lambda \I^{\muri,\nu^k}_k(\E; \T^k_{H+1}) \right] \notag \\
&= \sum_{k=1}^K \ex_{\D_k}\! \left[ 
\left( \max_{\nu'} \ex_{k}\! \left[V_1^{\E,*}(s_1) - V_{1,\muids, \nu'}^{\E} (s_1) \right] - \lambda \I^{\muids,\nu'}_k(\E; \T^k_{H+1}) \right)  + \lambda \I^{\muri,\nu^k}_k(\E; \T^k_{H+1}) \right]. \notag
\end{align}
 Next, due to the AM-GM inequality, we have  
\begin{align*}
&\ex_{k} \left[V_1^{\E,*}(s_1) - V_{1,\muids, \nu'}^{\E} (s_1) \right] \\
&\le \ex_{k} \left[V_{1,\mu^*(\E),\nu'}^{\E}(s_1) - V_{1,\muids, \nu'}^{\E} (s_1) \right] \\
&= \frac{\ex_{k} \left[V_{1,\mu^*(\E),\nu'}^{\E} (s_1) - V_{1,\muids, \nu'}^{\E} (s_1) \right]}{\sqrt{\lambda \I^{\muids,\nu'}_k(\E; \T^k_{H+1})}} \cdot \sqrt{\lambda \I^{\muids,\nu'}_k(\E; \T^k_{H+1})} \\
&\le \frac{\left(\ex_{k} \left[V_{1,\mu^*(\E),\nu'}^{\E}(s_1) - V_{1,\muids, \nu'}^{\E} (s_1) \right] \right)^2}{4\lambda \I^{\muids,\nu'}_k(\E; \T^k_{H+1})} + \lambda \I^{\muids,\nu'}_k(\E; \T^k_{H+1})\\
&= \frac{1}{4\lambda} \Gamma_k(\muids, \nu',\E) + \lambda \I^{\muids,\nu'}_k(\E; \T^k_{H+1}).
\end{align*}
Therefore, for the policy $(\murii,\nurii)$ of {\sc Reg-MAIDS}, we have
\begin{align}
&\ex \left(\sum_{k=1}^K V_1^{\E,*}(s_1) - V_{1,\muri, \nuri}^{\E} (s_1) \right) \\
&\le  \sum_{k=1}^K \ex_{\D_k}\! \left[\left(\max_{\nu'} \frac{\Gamma_k(\muids, \nu',\E)}{4\lambda} \right) + \lambda \I^{\muri,\nuri}_k(\E; \T^k_{H+1}) \right] \\
&\le  \sum_{k=1}^K \ex_{\D_k}\! \left[\left(\max_{\nu'} \frac{\Gamma_k(\muts, \nu',\E)}{4\lambda} \right)\right] + \lambda \I^{\murii,\nurii}(\E; \D_{K+1}) \label{eq:guo} \\
&\le K \cdot \frac{4H^3SAB}{4\lambda} + \lambda \cdot 2S^2ABH \log(SKH)
\label{eq:guo2} \\
&= \frac{5}{2} \sqrt{2H^4S^3A^2B^2K \log(SKH)}. \label{eq:guo3}
\end{align}
where Eqn.~\eqref{eq:guo} follows from the definition of $\muids$ in~\eqref{eq:ids1} as well as the derivations in~\eqref{eq:mi1}-\eqref{eq:mi3}. Eqn.~\eqref{eq:guo2} is obtained by recalling the upper bound on $\Gamma_k(\muts, \nu, \E)$  presented in Lemma~\ref{lemma:ts} and the upper bound on $\I^{\murii,\nu}(\E; \D_{K+1})$ presented in Lemma~\ref{lemma:mi}, while Eqn.~\eqref{eq:guo3}  is obtained by choosing $$\lambda = \sqrt{2KH^2/S\log(SKH)}.$$

We remark that the  second term in~\eqref{eq:ji} can be handled in a similar fashion. By choosing the parameter $$\tl = \lambda = \sqrt{2KH^2/S\log(SKH)},$$ one can show that the second term in~\eqref{eq:ji} is also upper bounded by $\frac{5}{2} \sqrt{2H^4S^3A^2B^2K \log(SKH)}$, leading to the following upper bound on the Bayesian regret:
\begin{align}
    \mathsf{BR}_K (\murii) &= \ex \left(\sum_{k=1}^K V_1^{\E,*}(s_1) - V_{1,\muri, \dagger}^{\E} (s_1) \right) \le  8H^2S^{3/2}AB\sqrt{K \log(SKH)}.
\end{align}

\section{Appendix for {\sc Compressed-MAIDS}} \label{sec:proof_thm3}

\subsection{Comparison of distortion measures}
\begin{remark}{\em
The distortion measure $d_{\Phi_A, \Phi_B}$ is similar to the one defined in~\cite{arumugam2022deciding} for MDPs, which follows the \emph{value equivalence principle}~\cite{grimm2020value,grimm2021proper}. We note that their definition focuses on the value difference through the Bellman update; specifically, for some value function class $\mathcal{V}$ and two environments $e$ and $e'$, they define the distortion between $e$ and $e'$ by $$\sup_{\mu \in \Phi_A }\sup_{V \in \mathcal{V}} \big|\ex_{a \sim \mu(\cdot|s), s' \sim P^e(\cdot|s,a)} [V(s')] - \ex_{ a\sim \mu(\cdot|s), s' \sim P^{e'}(\cdot|s,a)} [V(s')]\big|$$ in the context of time-homogeneous MDPs with only one player. Note that they require an additional value function class~$\mathcal{V}$; in contrast, we directly let the value function $V_{1,\mu,\nu}^e(s_1)$  depend on the policy and the environment. The similarity of these two measures (if considered in the context of zero-sum MGs), roughly speaking, is that our definition $|V_{1,\mu,\nu}^e(s_1) - V_{1,\mu,\nu}^{e'}(s_1)|$ can be seen as the \emph{cumulative} differences of values through the Bellman update for the $H$ steps.   }   \label{remark:distortion}
\end{remark}

\subsection{A concrete construction of the compressed environment } \label{appendix:F2}
Below, we provide a concrete construction of the compressed environment $\te$ that satisfies~\eqref{eq:constraint2} for any $\Phi_A \subseteq \Pi_A$ and $\Phi_B \subseteq \Pi_B$. 
First, we introduce the concept of \emph{covering number}.
\begin{definition}[covering number] \label{def:cover}
For any $\delta \ge 0$, we say the set $\mathcal{C}(\delta) \subseteq \Delta(\s)$ is a $\delta$-covering of the probability simplex $\Delta(\s)$ w.r.t. the total variation distance $\DD_{\text{\emph{TV}}}$ if 
$\forall P \in \Delta(\s), \ \exists P' \in \mathcal{C}(\delta) \text{ such that } \DD_{\text{\emph{TV}}}(P,P') \le \delta.$ 
The $\delta$-covering number of $\Delta(\s)$ is then denoted as 
\begin{align*}
&\kappa(\delta) \triangleq \min\{n: \exists \text{ a $\delta$-covering of } \Delta(\s) \text{ with cardinality }    n  \}.
\end{align*} 
\end{definition}

Next, we describe the construction of the subspaces $\{\Theta_c\}_{c=1}^{\C}$ (with $\C = \kappa(\epsilon/2H^2)^{SABH}$) as well as the construction of the compressed environment $\te$. 
\begin{enumerate}
\item We partition the environment space $\Theta$ into $\C \triangleq \kappa(\epsilon/2H^2)^{SABH}$ subspaces $\{\Theta_c \}_{c=1}^{\C}$, where $\kappa(\epsilon/2H^2)$ is the covering number. For each $(s,a,b,h) \in \s \times \A\times \B\times [H]$, we divide the support of $P_h(\cdot|s,a,b)$, which is $\Delta(\s)$, into $\kappa(\epsilon/2H^2)$ balls of radius $\epsilon/2H^2$ (w.r.t. the total variation distance). Each ball is centered at one element of the set $\mathcal{C}(\epsilon/2H^2)$, where $\mathcal{C}(\epsilon/2H^2)$ is an $(\epsilon/2H^2)$-covering of $\Delta(\s)$. Thus, for any two distributions $P$ and $\tilde{P}$ in the same ball, by the triangle inequality, it is clear that $$\DD_{\text{TV}}(P,\tilde{P}) \le \DD_{\text{TV}}(P,P') + \DD_{\text{TV}}(P',\tilde{P}) \le \epsilon/H^2$$ where $P'$ is the ball center. Next, we ``trim'' these balls in such a way that if a distribution $P \in \Delta(\s)$ appears in two or more balls, we remove it from all but one ball. This trimming operation ensures that the partitions of $\Delta(\s)$ are disjoint. Since we need to divide the support of $P_h(\cdot|s,a,b)$ for all $(s,a,b,h) \in \s \times \A\times \B\times [H]$, we have $\C = \kappa(\epsilon/2H^2)^{SABH}$. It is clear that the subspaces $\{\Theta_c \}_{c=1}^{\C}$ are disjoint and their union constitutes the space~$\Theta$. 

\item For each subspace $\Theta_c$, we arbitrarily select an element $e_c \in \Theta_c$ as the \emph{reference environment} of $\Theta_c$. We then construct the compressed environment $\te$ in such a way that 
\begin{align}
 \te = e_c \ \ \ \text{if and only if }  \ \E \in \Theta_c. \label{eq:concrete}   
\end{align}
\end{enumerate}

Lemma~\ref{lemma:property} below shows that the above construction of $\te$ satisfies the hard-compression constraint in~\eqref{eq:constraint2}.

\begin{lemma} \label{lemma:property}
For every episode $k \in [K]$ and any $\Phi_A \subseteq \Pi_A$ and $\Phi_B \subseteq \Pi_B$, the compressed environment $\te$ satisfies that $$\PP(d_{\Phi_A, \Phi_B}(\E,\te) > \epsilon ) = 0.$$
\end{lemma}

\begin{proof} \label{appendix:lemma2}

For any subspace $\Theta_c$, any pair of environment $e,e' \in \Theta_c$, and any policy $(\mu,\nu) \in \Phi_A \times \Phi_B$, we have
\begin{align}
&V_{1,\mu,\nu}^e(s_1) - V_{1,\mu,\nu}^{e'}(s_1) \\
&=\sum_{h=1}^H \ex_{\mu,\nu}^{e'}\left[ \ex_{s' \sim P_h^e(\cdot|s_h,a_h,b_h)} [V_{h+1,\mu,\nu}^e(s')] - \ex_{s' \sim P_h^{e'}(\cdot|s_h,a_h,b_h)} [V_{h+1,\mu,\nu}^e(s') \right]  \label{eq:sss}\\
&=\sum_{h=1}^H \ex_{\mu,\nu}^{e'} \left[[P_h^{e}(\cdot|s_h,a_h,b_h) - P_h^{e'}(\cdot|s_h,a_h,b_h)] \cdot V_{h+1,\mu,\nu}^e(s')  \right] \\
&\le H \sum_{h=1}^H \ex_{\mu,\nu}^{e'} \left[ \big|P_h^{e}(\cdot|s_h,a_h,b_h) - P_h^{e'}(\cdot|s_h,a_h,b_h)\big| \right] \\
&\le H \sum_{h=1}^H \max_{s,a,b} \big|P_h^{e}(\cdot|s,a,b) - P_h^{e'}(\cdot|s,a,b)\big| \\
&\le \epsilon,
\end{align}
where~\eqref{eq:sss} follows from Lemma~\ref{lemma:diff}, and the last inequality follows from the property of $\Theta_c$. Therefore, based on the construction of $\te$ in~\eqref{eq:concrete}, it is clear that 
\begin{align}
 \PP\left(d_{\Phi_A, \Phi_B}(\E,\te) > \epsilon \right) = 0.
    \end{align}
This completes the proof of Lemma~\ref{lemma:property}.
\end{proof}

\subsection{Proof of Theorem~\ref{thm:compressed}}  Recall that 
\begin{align}
\mathsf{BR}_K (\mucii) &= \ex \left(\sum_{k=1}^K V_1^{\E,*}(s_1) - V_{1,\muci, \dagger}^{\E} (s_1) \right) \\
&= \ex \left(\sum_{k=1}^K V_1^{\E,*}(s_1) - V_{1,\muci, \nuci}^{\E} (s_1) \right) + \ex \left(\sum_{k=1}^K V_{1,\muci, \nuci}^{\E} (s_1) - V_{1,\muci, \dagger}^{\E} (s_1) \right). \label{eq:jii}
\end{align} 
For the first term in~\eqref{eq:jii}, we have 
\begin{align}
&\ex \left(\sum_{k=1}^K V_1^{\E,*}(s_1) - V_{1,\muci, \nuci}^{\E} (s_1) \right) \label{eq:q18} \\
&\le \ex \left(\sum_{k=1}^K V_{1,\mu^*(\E),\nuci}^{\E}(s_1) - V_{1,\muci, \nuci}^{\E} (s_1) \right) \label{eq:qnew} \\
&= \sum_{k=1}^K \ex_{\D_k} \left[ \ex_{k} \left(V_{1,\mu^*(\E),\nuci}^{\E}(s_1) - V_{1,\muci, \nuci}^{\E} (s_1) \right) - 2\epsilon \right] + 2K \epsilon. \label{eq:64}
\end{align}
In Lemma~\ref{lemma:property2} below, we characterize the difference of performance of the same policy on the environment $\E$ and the compressed environment $\te$.  
\begin{lemma} \label{lemma:property2}
    For any episode $k \in [K]$ and any policy $(\mu,\nu)$, we have
    \begin{align}
\ex_k\left[ V_{1,\mu^*(\E),\nu}^{\E} (s_1) - V_{1,\mu,\nu}^{\E} (s_1) \right] - 2\epsilon \le \ex_k\left[ V_{1,\mu^*(\E),\nu}^{\te} (s_1) - V_{1,\mu,\nu}^{\te} (s_1) \right]. 
    \end{align}
\end{lemma}

\begin{proof}[Proof of Lemma~\ref{lemma:property2}]
Recall that the compressed environment $\te$ constructed in~\eqref{eq:concrete} satisfies the hard-compression constraint for any $\Phi_A \subseteq \Pi_A$ and $\Phi_B \subseteq \Pi_B$. Setting $\Phi_A = \Pi_A$ and $\Phi_B = \Pi_B$, we have
\begin{align}
&\ex_k\left[ V_{1,\mu^*(\E),\nu}^{\E} (s_1)\right] - \ex_k\left[ V_{1,\mu^*(\E),\nu}^{\te} (s_1)\right] \\
&= \PP_k( d_{\Phi_A,\Phi_B}(\E,\te) > \epsilon ) \cdot \ex_k\left[V_{1,\mu^*(\E),\nu}^{\E} (s_1) - V_{1,\mu^*(\E),\nu}^{\te} (s_1) \big| d_{\Phi_A,\Phi_B}(\E,\te) > \epsilon \right] \notag \\ 
&\qquad + \PP_k( d_{\Phi_A,\Phi_B}(\E,\te) \le \epsilon ) \cdot \ex_k\left[V_{1,\mu^*(\E),\nu}^{\E} (s_1) - V_{1,\mu^*(\E),\nu}^{\te} (s_1) \big| d_{\Phi_A,\Phi_B}(\E,\te) \le \epsilon \right] \\
&= \ex_k\left[V_{1,\mu^*(\E),\nu}^{\E} (s_1) - V_{1,\mu^*(\E),\nu}^{\te} (s_1) \big| d_{\Phi_A,\Phi_B}(\E,\te) \le \epsilon \right] \label{eq:bai1} \\
&\le \epsilon, \label{eq:bai2}
\end{align}
where Eqn.~\eqref{eq:bai1} is due to the fact that the compressed environment $\te$ satisfies $\PP_k( d_{\Phi_A,\Phi_B}(\E,\te) > \epsilon ) = 0$, and Eqn.~\eqref{eq:bai2} follows from the definition of $d_{\Phi_A,\Phi_B}(\E,\te)$ for $\Phi_A = \Pi_A$ and $\Phi_B = \Pi_B$.
Similarly, one can also show that 
\begin{align}
\ex_k\left[ V_{1,\mu,\nu}^{\E} (s_1)\right] - \ex_k\left[ V_{1,\mu,\nu}^{\E} (s_1)\right] \le \epsilon.  \label{eq:bai3}  
\end{align}
Combining~\eqref{eq:bai2} and~\eqref{eq:bai3} together, we complete the proof of Lemma~\ref{lemma:property2}.
\end{proof}

Based on Lemma~\ref{lemma:property2}, we can bound~\eqref{eq:64} from above as
\begin{align}
&\sum_{k=1}^K \ex_{\D_k} \left[ \ex_{k} \left(V_{1,\mu^*(\E),\nuci}^{\E}(s_1) - V_{1,\muci, \nuci}^{\E} (s_1) \right) - 2\epsilon \right] + 2K \epsilon \\
&\le \sum_{k=1}^K \ex_{\D_k} \left[ \ex_{k} \left(V_{1,\mu^*(\E),\nuci}^{\te}(s_1) - V_{1,\muci, \nuci}^{\te} (s_1) \right) \right] + 2K \epsilon \label{eq:kemm} \\
&\le\sum_{k=1}^K \ex_{\D_k} \left[\sqrt{\tilde{\Gamma}_k(\muci,\nuci,\te)} \sqrt{\I_k^{\muci,\nuci}(\te;\T^{k}_{H+1})} \right] + 2K \epsilon \\
&\le \sqrt{\sum_{k=1}^K \ex_{\D_k}\left[\tilde{\Gamma}_k(\muci,\nuci,\te) \right] } \cdot \sqrt{\sum_{k=1}^K \ex_{\D_k} \left[\I_k^{\muci,\nuci}(\te;\T^{k}_{H+1}) \right]} +2K \epsilon \label{eq:qcauchy} \\
&\le \sqrt{\sum_{k=1}^K \ex_{\D_k}\left[\max_{\nu} \tilde{\Gamma}_k(\muci,\nu,\te) \right] } \cdot \sqrt{\I^{\mucii,\nucii}(\te; \D_{K+1}) } + 2K \epsilon \\
&\le \sqrt{\sum_{k=1}^K \ex_{\D_k}\left[\max_{\nu} \tilde{\Gamma}_k(\muts,\nu,\te) \right] } \cdot \sqrt{\I^{\mucii,\nucii}(\te; \D_{K+1}) } + 2K \epsilon.  \label{eq:q37}
\end{align}
where~\eqref{eq:qcauchy} follows from the Cauchy-Schwarz inequality, and~\eqref{eq:q37} is due to the definition of the max-player's policy $\mucii$, as well as the chain rule of mutual information. Using the same proof technique as for Lemma~\ref{lemma:ts}, one can prove that for any $\nu$ and any distribution of $\E$, 
\begin{align}
    \tilde{\Gamma}_k(\muts,\nu,\te)\le 4H^3SAB. 
\end{align}
This means that the first term in~\eqref{eq:jii} satisfies
\begin{align}
    \ex \left(\sum_{k=1}^K V_1^{\E,*}(s_1) - V_{1,\muci, \nuci}^{\E} (s_1) \right) \le  \sqrt{4KH^3SAB \cdot \I^{\mucii,\nucii}(\te; \D_{K+1})} + 2K\epsilon.
\end{align}

For the second term in~\eqref{eq:jii}, we have  
\begin{align}
&\ex\left(\sum_{k=1}^K V_{1,\muci, \nuci}^{\E} (s_1) - V_{1,\muci, \dagger}^{\E} (s_1) \right) \\
&= \sum_{k=1}^K \ex_{\D_k} \left[ \ex_k \left( V_{1,\muci, \nuci}^{\E} (s_1) - V_{1,\muci, \dagger}^{\E} (s_1) \right) - 2\epsilon \right] + 2K\epsilon. \label{eq:45}
\end{align}
Similar to Lemma~\ref{lemma:property2}, we can also show that
\begin{align}
&\ex_k\left[ V_{1,\muci,\nuci}^{\E} (s_1) - V_{1,\muci,\nu^{\dagger}_{\E}(\muci)}^{\E} (s_1) \right] - 2\epsilon \notag \\
&\qquad\qquad \le \ex_k\left[ V_{1,\muci,\nuci}^{\te} (s_1) - V_{1,\muci,\nu^{\dagger}_{\te}(\muci)}^{\te} (s_1) \right]. 
    \end{align}
Thus,~\eqref{eq:45} can be bounded from above as 
\begin{align}
&\sum_{k=1}^K \ex_{\D_k} \left[ \ex_k \left( V_{1,\muci,\nuci}^{\te} (s_1) - V_{1,\muci,\nu^{\dagger}_{\te}(\muci)}^{\te} (s_1) \right)  \right] + 2K\epsilon  \\
&\le \sum_{k=1}^K \ex_{\D_k} \left[ \sqrt{\tilde{\Lambda}_k^{\muci}(\nuci,\te) } \times \sqrt{\I_k^{\muci,\nuci}(\te; \T_{H+1}^k) }  \right] + 2K\epsilon \\
&\le \sqrt{\sum_{k=1}^K \ex_{\D_k}\left[ \tilde{\Lambda}_k^{\muci}(\nuci,\te)  \right]  } \times  \sqrt{\sum_{k=1}^K \ex_{\D_k}\left[\I_k^{\muci,\nuci}(\te; \T_{H+1}^k) \right]}+ 2K\epsilon \\
&\le \sqrt{\sum_{k=1}^K \ex_{\D_k}\left[ \tilde{\Lambda}_k^{\muci}(\nuts(\muci),\te)  \right]  } \times  \sqrt{\I^{\mucii,\nucii}(\te; \D_{K+1}) }+ 2K\epsilon, \label{eq:50}
\end{align}
where $\nuts(\muci)$ in~\eqref{eq:50} is the TS policy of the min-player with respect to the max-player's policy $\muci$; specifically, it first samples a realization of the environment $\E = \e$, and then chooses the best response to $\muci$ under environment $\e$. Using the same proof technique as for Lemma~\ref{lemma:ts}, we have $$\tilde{\Lambda}_k^{\muci}(\nuts(\muci),\te) \le 4H^3SAB$$ 
for any distribution of $\E$.  Thus,  the second term in~\eqref{eq:jii} satisfies
\begin{align}
    \ex\left(\sum_{k=1}^K V_{1,\muci, \nuci}^{\E} (s_1) - V_{1,\muci, \dagger}^{\E} (s_1) \right) \le  \sqrt{4K H^3SAB \cdot \I^{\mucii,\nucii}(\te; \D_{K+1})} + 2K\epsilon.
\end{align}
Combining the upper bounds for the first term and the second term in~\eqref{eq:jii} together, we eventually obtain that 
\begin{align}
\mathsf{BR}_K (\mucii) \le  4 \sqrt{KH^3 SAB \cdot \I^{\mucii,\nucii}(\te; \D_{K+1}) } + 4K\epsilon. 
\end{align}
This completes the proof of Theorem~\ref{thm:compressed}.

\section{Proof of Lemma~\ref{lemma:eq}} \label{sec:proof_eq}
By recalling the definition of $\trh$ in Section~\ref{sec:regularized}, we have
\begin{align}
&\ex_{\mu,\nu}^{\me} \left[\sum_{h=1}^H \trh(s_h,a_h,b_h) \right] \\
&= \ex_{\mu,\nu}^{\me} \left[\sum_{h=1}^H r_h(s_h,a_h,b_h) + \lambda \ex_k\left[ \DD_{\text{KL}}(P_h^{\E}(\cdot|s_h,a_h,b_h) \Vert P_h^{\me}(\cdot|s_h,a_h,b_h)) \right] \right] \\
&= \sum_{h=1}^H \sum_{s,a,b} d_{h,\mu,\nu}^{\me}(s,a,b) \Big[r_h(s,a,b) + \lambda \ex_k\left[ \DD_{\text{KL}}(P_h^{\E}(\cdot|s,a,b) \Vert P_h^{\me}(\cdot|s,a,b)) \right] \Big]\\
&= \sum_{h=1}^H \sum_{s,a,b} \ex_k[d_{h,\mu,\nu}^{\E}(s,a,b)] \cdot r_h(s,a,b) \notag \\
&\qquad \qquad \qquad + \lambda \ex_k\left[ \sum_{h=1}^H \sum_{s,a,b} d_{h,\mu,\nu}^{\me}(s,a,b) \DD_{\text{KL}}(P_h^{\E}(\cdot|s,a,b) \Vert P_h^{\me}(\cdot|s,a,b)) \right]  \label{eq:heihei1}\\
&= \ex_k\left[ V_{1,\mu,\nu}^{\E}(s_1) \right] + \lambda \I_k^{\mu,\nu}(\E; \T^k_{H+1}) \label{eq:heihei2}
\end{align}
where~\eqref{eq:heihei1} uses the fact that $d_{h,\mu,\nu}^{\me}(s,a,b) = \ex_k[d_{h,\mu,\nu}^{\E}(s,a,b)]$ and~\eqref{eq:heihei2} follows from Lemma~\ref{lemma:mikl}.

\section{Proof of Lemma~\ref{lemma:ts}} \label{sec:proof_ts}

Recall that 
\begin{align}
\Gamma_k(\muts,\nu,\E) = \frac{\left(\ex_k\left[V_{1,\mu^*(\E),\nu}^{\E}(s_1) - V_{1,\muts,\nu}^{\E}(s_1) \right] \right)^2}{\I_k^{\muts,\nu}(\E; \T^{k}_{H+1})}. \label{eq:meng}
\end{align}
Let us focus on the numerator of the RHS of~\eqref{eq:meng} from now on. It is worth noting that the term $\ex_k[V_{1,\mu^*(\E),\nu}^{\E}(s_1) - V_{1,\muts,\nu}^{\E}(s_1) ]$ may not always be positive; thus, to provide an upper bound on the numerator of the RHS of~\eqref{eq:meng}, one need to upper-bound the absolute value of $\ex_k [V_{1,\mu^*(\E),\nu}^{\E}(s_1) - V_{1,\muts,\nu}^{\E}(s_1) ]$. 

Given the past trajectory $\D_k$, we define an environment $\met = (H,\s,\A,\B,\{P_h^{\met}\}_{h=1}^H, \{r_h\}_{h=1}^H)$, where the transition kernel
$$P_h^{\met}(\cdot | s,a,b) \triangleq \ex_{\E \sim \PP(\cdot|\D_k)}[P_h^{\E}(\cdot | s,a,b)].$$ Note that $\met$ is similar to the mean environments $\me$ and $\me'$ defined in Section~\ref{sec:regularized}, but differs from them in terms of the reward functions. Based on $\met$, we now decompose the absolute value of $\ex_k [V_{1,\mu^*(\E),\nu}^{\E}(s_1) - V_{1,\muts,\nu}^{\E}(s_1) ]$ as follows: 
\begin{align}
 & \left|\ex_k\left[V_{1,\mu^*(\E),\nu}^{\E}(s_1) - V_{1,\muts,\nu}^{\E}(s_1) \right] \right| \notag \\
 &\le \left| \ex_k\left[V_{1,\mu^*(\E),\nu}^{\E}(s_1) - V_{1,\muts,\nu}^{\met}(s_1) \right] \right| + \left| \ex_k\left[V_{1,\muts,\nu}^{\met}(s_1) - V_{1,\muts,\nu}^{\E}(s_1) \right] \right|.   \label{eq:dec}
\end{align}
For the first term in~\eqref{eq:dec}, we have 
\begin{align}
&\left|\ex_k\left[V_{1,\mu^*(\E),\nu}^{\E}(s_1) - V_{1,\muts,\nu}^{\met}(s_1) \right] \right| \\
&= \left| \ex_k\left[V_{1,\mu^*(\E),\nu}^{\E}(s_1)\right] - V_{1,\muts,\nu}^{\met}(s_1) \right|  \label{eq:hu1}  \\
&= \left| \ex_k\left[V_{1,\mu^*(\E),\nu}^{\E}(s_1)\right] - \ex_{\E \sim \PP(\cdot|\D_k)} \left[ V_{1,\mu^{*}(\E),\nu}^{\met}(s_1) \right] \right| \label{eq:hu2} \\
&= \left|\ex_k\left[V^{\E}_{1,\mu^*(\E),\nu}(s_1) -  V_{1,\mu^{*}(\E),\nu}^{\met}(s_1) \right] \right|, \label{eq:hu4}
\end{align}
where Eqn.~\eqref{eq:hu1} holds since the value of $V_{1,\muts,\nu}^{\met}(s_1)$  is independent of the environment $\E$, Eqn.~\eqref{eq:hu2} is due to the rule of the Thompson sampling policy $\muts$, Eqn.~\eqref{eq:hu4} follows from the fact that the notations  $\ex_k(\cdot)$ and $\ex_{\E \sim \PP(\cdot|\D_k)}(\cdot)$ are equivalent.   Below, we introduce a lemma that characterizes the difference of performance of the same policy on two different environments.
\begin{lemma} \label{lemma:diff}
    Consider any two environments $\e$ and $\e'$ with different transition kernels $P_{h}^e$ and $P_{h}^{e'}$ but the same reward functions $\{r_h\}_{h=1}^H$. For any fixed policy $(\mu,\nu)$, we have 
\begin{align}
&V_{1,\mu,\nu}^e(s_1) - V_{1,\mu,\nu}^{e'}(s_1) \\
&= \sum_{h=1}^H \ex_{\mu,\nu}^{e'}\left[ \ex_{s' \sim P_h^e(\cdot|s_h,a_h,b_h)} [V_{h+1,\mu,\nu}^e(s')] - \ex_{s' \sim P_h^{e'}(\cdot|s_h,a_h,b_h)} [V_{h+1,\mu,\nu}^e(s') \right]    \\
&= \sum_{h=1}^H \ex_{\mu,\nu}^{e}\left[ \ex_{s' \sim P_h^e(\cdot|s_h,a_h,b_h)} [V_{h+1,\mu,\nu}^{e'}(s')] - \ex_{s' \sim P_h^{e'}(\cdot|s_h,a_h,b_h)} [V_{h+1,\mu,\nu}^{e'}(s') \right].
\end{align}
\end{lemma}

\begin{proof}[Proof of Lemma~\ref{lemma:diff}]
The proof is adapted from~\cite[Lemma~D.3]{hao2022regret} with appropriate modifications for two-player zero-sum MGs.   
\end{proof} 

For notational convenience, we define
\begin{align}
    \Delta(\E,\met,s,a,b) \triangleq \ex_{s' \sim P_h^{\E}(\cdot|s_h,a_h,b_h)} [V_{h+1,\mu^*(\E),\nu}^{\E}(s')] - \ex_{s' \sim P_h^{\met}(\cdot|s_h,a_h,b_h)} [V_{h+1,\mu^*(\E),\nu}^{\E}(s')]
\end{align}
and define the \emph{occupancy measure} with respect to any policy $(\mu,\nu)$ and environment $\e$  as
\begin{align}
d_{h,\mu,\nu}^{\e}(s,a,b) \triangleq \PP_{\mu,\nu}^{\e}(s_h^k = s, a_h^k = a, b_h^k = b).
\end{align}

Applying Lemma~\ref{lemma:diff} to Eqn.~\eqref{eq:hu4}, we have
\begin{align}
&\left| \ex_k\left[V^{\E}_{1,\mu^*(\E),\nu}(s_1) -  V_{1,\mu^{*}(\E),\nu}^{\met}(s_1) \right] \right| \\ 
&=\left| \ex_k \left[ \sum_{h=1}^H \ex_{\mu^{*}(\E),\nu}^{\met} \left[\Delta(\E,\met,s_h^k,a_h^k,b_h^k) \right]\right] \right| \\
&\le \sum_{h=1}^H \ex_k \left[ \sum_{s,a,b}  d_{h,\mu^{*}(\E),\nu}^{\met}(s,a,b) \cdot \big|\Delta(\E,\met,s,a,b)\big| \right]. \label{eq:xi}
\end{align}
A key observation is that the expected occupancy measure $\ex_k[d_{h,\mu^{*}(\E),\nu}^{\met}(s,a,b)]$  with respect to the random environment $\E$ is equal to the occupancy measure of the Thompson sampling policy $d_{h,\muts,\nu}^{\met}(s,a,b)$, due to the definition of $\muts$. Thus, one can further express Eqn.~\eqref{eq:xi} as
\begin{align}
&\sum_{h=1}^H \ex_k \left[\sum_{s,a,b} d_{h,\mu^{*}(\E),\nu}^{\met}(s,a,b) \cdot \big|\Delta(\E,\met,s,a,b) \big| \right] \\
&= \sum_{h=1}^H \ex_k \left[\sum_{s,a,b} \frac{d_{h,\mu^{*}(\E),\nu}^{\met}(s,a,b)}{\sqrt{\ex_k[d_{h,\mu^{*}(\E),\nu}^{\met}(s,a,b)]}} \cdot \sqrt{d_{h,\muts,\nu}^{\met}(s,a,b)} \cdot \big|\Delta(\E,\met,s,a,b)\big| \right] \\
&\le \sqrt{\sum_{h=1}^H \ex_k\left[\sum_{s,a,b} \frac{\left[d_{h,\mu^{*}(\E),\nu}^{\met}(s,a,b)\right]^2}{\ex_k[d_{h,\mu^{*}(\E),\nu}^{\met}(s,a,b)]} \right] } \cdot \sqrt{\sum_{h=1}^H \ex_k\left[\sum_{s,a,b} d_{h,\muts,\nu}^{\met}(s,a,b) \Delta(\E,\met,s,a,b)^2 \right]} \label{eq:ha}\\
&\le \sqrt{SABH} \cdot \sqrt{\sum_{h=1}^H \ex_k\left[\sum_{s,a,b} d_{h,\muts,\nu}^{\met}(s,a,b) \Delta(\E,\met,s,a,b)^2 \right]}, \label{eq:xi2}
\end{align}
where~\eqref{eq:ha} follows from the Cauchy-Schwarz inequality, and~\eqref{eq:xi2} is due to the linearity of expectation as well as the fact that the occupancy measure always satisfies  $$[d_{h,\mu^{*}(\E),\nu}^{\met}(s,a,b)]^2 \le d_{h,\mu^{*}(\E),\nu}^{\met}(s,a,b).$$  For the second part of~\eqref{eq:xi2}, we have  
\begin{align}
&\sum_{h=1}^H \ex_k\left[\sum_{s,a,b} d_{h,\muts,\nu}^{\met}(s,a,b) \Delta(\E,\met,s,a,b)^2 \right] \label{eq:old1} \\
&= \sum_{h=1}^H \ex_k\left[\sum_{s,a,b} d_{h,\muts,\nu}^{\met}(s,a,b) \left[\sum_{s'}\left(P_h^{\E}(s'|s,a,b) - P_h^{\met}(s'|s,a,b)\right) \cdot V_{h+1,\mu^*(\E),\nu}^{\E}(s')    \right]^2  \right] \\
&\le H^2 \sum_{h=1}^H \ex_k\left[\sum_{s,a,b} d_{h,\muts,\nu}^{\met}(s,a,b) \left[\sum_{s'}\left|P_h^{\E}(s'|s,a,b) - P_h^{\met}(s'|s,a,b)\right|    \right]^2  \right] \label{eq:tan1} \\
&= H^2 \sum_{h=1}^H \ex_k\left[\sum_{s,a,b} d_{h,\muts,\nu}^{\met}(s,a,b) \left[ \DD_{\text{TV}}\left(P_h^{\E}(\cdot|s,a,b)\Vert P_h^{\met}(\cdot|s,a,b)\right)    \right]^2  \right] \label{eq:tan2} \\
&\le H^2 \sum_{h=1}^H \ex_k\left[\sum_{s,a,b} d_{h,\muts,\nu}^{\met}(s,a,b) \frac{1}{2} \DD_{\text{KL}}\left(P_h^{\E}(\cdot|s,a,b)\Vert P_h^{\met}(\cdot|s,a,b)\right) \right], \label{eq:tan3}
\end{align}
where Eqn.~\eqref{eq:tan1} holds since the value function is always bounded from above by $H$, Eqn.~\eqref{eq:tan2} follows from the definition of the \emph{total variation distance} $\DD_{\text{TV}}(\cdot \Vert \cdot)$ between two distributions, and in Eqn.~\eqref{eq:tan3} we apply the Pinsker's inequality to relate the total variation distance to the KL-divergence between two distributions. Finally, we introduce a lemma (adapted from~\cite[Lemma~A.1]{hao2022regret}) showing that the term in~\eqref{eq:tan3} can be related to the mutual information between the environment $\E$ and the trajectory of episode $k$.
\begin{lemma} \label{lemma:mikl}
    For any policy $(\mu,\nu)$, we have 
    \begin{align}
    \I^{\mu,\nu}_k(\E; \T^{k}_{H+1}) = \sum_{h=1}^H \ex_k\left[\sum_{s,a,b} d_{h,\mu,\nu}^{\met}(s,a,b) \cdot \DD_{\text{KL}}\left(P_h^{\E}(\cdot|s,a,b) \Vert  P_h^{\met}(\cdot|s,a,b)\right) \right].    
    \end{align}
\end{lemma}

\begin{proof}[Proof of Lemma~\ref{lemma:mikl}]

Recall that $\T^{k}_{H+1}$ is the trajectory at episode $k$. We further define $$\T^{k}_{h} \triangleq \{s_1^k, a_1^k, b_1^k, r_1^k, \ldots, s_h^k, a_h^k, b_h^k, r_h^k\}$$ as the \emph{trajectory} at episode $k$ from steps $1$ to $h$  (for any $h \in [H]$). 
By the chain rule of mutual information, we have 
\begin{align}
&\I^{\mu,\nu}_k(\E; \T^{k}_{H+1}) \\
&= \sum_{h=1}^H \I^{\mu,\nu}_k(\E; (s_h^k,a_h^k,b_h^k,r_h^k)|\T^{k}_{h-1}) + \I^{\mu,\nu}_k(\E; s_{H+1}^k|\T^{k}_{H}) \\
&= \sum_{h=1}^{H+1} \I^{\mu,\nu}_k(\E; s_h^k|\T^{k}_{h-1}) + \sum_{h=1}^H \left[\I^{\mu,\nu}_k(\E; (a_h^k,b_h^k)|\T^{k}_{h-1},s_h^k) + \I^{\mu,\nu}_k(\E; r_h^k|\T^{k}_{h-1},s_h^k,a_h^k,b_h^k) \right]. \label{eq:56}
\end{align}
Note that $\I^{\mu,\nu}_k(\E; r_h^k|\T^{k}_{h-1},s_h^k,a_h^k,b_h^k)  = 0$ for all $h \in [H]$
since $r_h^k$ is a deterministic function of $(s_h^k,a_h^k,b_h^k)$. Moreover, one can also show that  $\I^{\mu,\nu}_k(\E; (a_h^k,b_h^k)|\T^{k}_{h-1},s_h^k) = 0$ for all $h \in [H]$,  since
\begin{align}
\I^{\mu,\nu}_k(\E; (a_h^k,b_h^k)|\T^{k}_{h-1},s_h^k) = \ex_k \left[\DD_{\text{KL}}\left(P(a_h^k,b_h^k|\E, \T^{k}_{h-1},s_h^k) \Vert P(a_h^k,b_h^k|\T^{k}_{h-1},s_h^k) \right) \right],    
\end{align}
and the fact that the actions $(a_h^k,b_h^k)$ only depend on the policy $(\mu,\nu)$ and the current state $s_h^k$ (which implies the KL-divergence term equals zero). Then, it suffices to focus on the first term of~\eqref{eq:56}. For each step $2 \le h \le H+1$, we have
\begin{align}
&\I^{\mu,\nu}_k(\E; s_h^k|\T^{k}_{h-1}) \\
&= \ex_{\T^{k}_{h-1}}\left[ \ex_{\E \sim P(\cdot|\D_k,\T^{k}_{h-1})} \left[ \DD_{\text{KL}}\left(P(s_h^k|\E,\T_{h-1}^k,\D_k) \Vert P(s_h^k|\T_{h-1}^k,\D_k)  \right) \right] \right] \\
&=\ex_{\T^{k}_{h-1}}\left[ \ex_{\E \sim \PP(\cdot|\D_k)} \left[ \DD_{\text{KL}}\left(P(s_h^k|\E,\T_{h-1}^k,\D_k) \Vert P(s_h^k|\T_{h-1}^k,\D_k)  \right) \right] \right] \label{eq:se}\\
&=\ex_{\T^{k}_{h-1}}\left[ \ex_{\E \sim \PP(\cdot|\D_k)} \left[ \DD_{\text{KL}}\left(P_{h-1}^{\E}(\cdot|s_{h-1}^k,a_{h-1}^k,b_{h-1}^k) \Vert P_{h-1}^{\met}(\cdot|s_{h-1}^k,a_{h-1}^k,b_{h-1}^k)  \right) \right] \right] \label{eq:se2} \\
&= \sum_{s,a,b} \PP_{\mu,\nu}^{\met}(s_{h-1}^k=s,a_{h-1}^k=a,b_{h-1}^k=b)\cdot  \ex_{k} \left[ \DD_{\text{KL}}\left(P_{h-1}^{\E}(\cdot|s,a,b) \Vert P_{h-1}^{\met}(\cdot|s,a,b)  \right) \right]  \label{eq:se3}\\
&= \ex_{k}\left[ \sum_{s,a,b} d_{h-1,\mu,\nu}^{\met}(s,a,b) \cdot \DD_{\text{KL}}\left(P_{h-1}^{\E}(\cdot|s,a,b) \Vert P_{h-1}^{\met}(\cdot|s,a,b)  \right) \right],
\end{align}
where~\eqref{eq:se} is due to the fact that the prior distribution on $\E$ is a product distribution over the $H$ steps, so that the trajectory $\T_{h-1}^k$  does not affect the posterior distribution of $\E$ with respect to the transition kernel $P_{h-1}$ appeared in the KL-divergence term. Eqn.~\eqref{eq:se2} holds since (i) the distribution of $s_h$ only depends on $(s_{h-1}^k,a_{h-1}^k,b_{h-1}^k)$ and the environment; and (ii) the definition of the mean environment $\met$ ensures that
\begin{align}
P(s_h^k|\T_{h-1}^k,\D_k) &= P(s_h^k|s_{h-1}^k,a_{h-1}^k,b_{h-1}^k,\D_k) \notag \\
&= \ex_{\E \sim \PP(\cdot|\D_k)} \left[ P(s_h^k|s_{h-1}^k,a_{h-1}^k,b_{h-1}^k,\D_k,\E ) \right]    \notag \\
&= P_{h-1}^{\met}(\cdot|s_{h-1}^k,a_{h-1}^k,b_{h-1}^k).
\end{align}
Eqn.~\eqref{eq:se3} is due to the fact that
\begin{align}
    \mathbb{E}_k \big(\mathbb{P}_{\mu,\nu}^{\E}(s_{h-1}^k =s, a_{h-1}^k =a, b_{h-1}^k =b ) \big) =  \mathbb{P}_{\mu,\nu}^{\met}(s_{h-1}^k =s, a_{h-1}^k =a, b_{h-1}^k =b ).
\end{align}
Also note that $\I^{\mu,\nu}_k(\E; s_h^k|\T^{k}_{h-1}) = 0$ when $h=1$, as the state $s_1^k = s_1$ is deterministic. Therefore, we have 
\begin{align}
    \I^{\mu,\nu}_k(\E; \T^{k}_{H+1}) = \sum_{h=1}^H \ex_{k}\left[ \sum_{s,a,b} d_{h,\mu,\nu}^{\met}(s,a,b) \cdot \DD_{\text{KL}}\left(P_{h}^{\E}(\cdot|s,a,b) \Vert P_{h}^{\met}(\cdot|s,a,b)  \right) \right].
\end{align}
This completes the proof of Lemma~\ref{lemma:mikl}.

\end{proof}

Applying Lemma~\ref{lemma:mikl}, we have
\begin{align}
\left| \ex_k\left[V_{1,\mu^*(\E),\nu}^{\E}(s_1) - V_{1,\muts,\nu}^{\met}(s_1) \right] \right| &\le \sqrt{SABH} \cdot \sqrt{H^2 \I_k^{\muts,\nu}(\E; \T^{k}_{H})} \\
    &= \sqrt{SABH^3 \I_k^{\muts,\nu}(\E; \T^{k}_{H})}.
\end{align}

For the second term in the RHS of~\eqref{eq:dec}, we have
\begin{align}
&\left|\ex_k\left[V_{1,\muts,\nu}^{\met}(s_1) - V_{1,\muts,\nu}^{\E}(s_1) \right]\right| \label{eq:ni1} \\
&= \left| \ex_k \left[ \sum_{h=1}^H \ex_{\muts,\nu}^{\met} \left[ \sum_{s'}P_h^{\met}(s'|s_h^{k},a_h^{k},b_h^{k}) V^{\E}_{h+1,\muts,\nu}(s') - \sum_{s'}P_h^{\E}(s'|s_h^{k},a_h^{k},b_h^{k}) V^{\E}_{h+1,\muts,\nu}(s') \right] \right] \right| \label{eq:ni2}\\
&\le H \sum_{h=1}^H \ex_k \left[  \sum_{s,a,b} d^{\met}_{h,\muts,\nu}(s,a,b) \left[ \sum_{s'} \left| P_h^{\met}(s'|s,a,b) - P_h^{\E}(s'|s,a,b) \right| \right] \right] \label{eq:ni4} \\
&\le H \sum_{h=1}^H \ex_k \left[  \sum_{s,a,b} d^{\met}_{h,\muts,\nu}(s,a,b) \sqrt{\frac{1}{2}\DD_{\text{KL}}\left(P_h^{\E}(\cdot|s,a,b) \Vert P_h^{\met}(\cdot|s,a,b)\right)  }  \right] \label{eq:ni5}\\
&\le H^2 \sqrt{\frac{1}{2H} \sum_{h=1}^H \ex_k \left[ \sum_{s,a,b} d^{\met}_{h,\muts,\nu}(s,a,b) \cdot \DD_{\text{KL}}\left(P_h^{\E}(\cdot|s,a,b)\Vert P_h^{\met}(\cdot|s,a,b)\right) \right]  } \label{eq:ni6} \\
&= \sqrt{\frac{H^3}{2} \I_k^{\muts,\nu}(\E; \T^{k}_{H+1}) }, \label{eq:ni7}
\end{align}
where~\eqref{eq:ni2} is due to Lemma~\ref{lemma:diff},~\eqref{eq:ni4} is due to the fact that the value function is always bounded from above by $H$,~\eqref{eq:ni5} follows from the Pinsker's inequality,~\eqref{eq:ni6} follows from Jensen's inequality, and~\eqref{eq:ni7} follows from Lemma~\ref{lemma:mikl}.

Therefore, we obtain
\begin{align}
    \Gamma_k(\muts,\nu,\E) &= \frac{\left(\ex_k\left[V_{1,\mu^*(\E),\nu}^{\E}(s_1) - V_{1,\muts,\nu}^{\E}(s_1) \right] \right)^2}{\I_k^{\muts,\nu}(\E; \T^{k}_{H+1})}\\
    &\le \frac{\left( \left| \ex_k\left[V_1^{\E,*}(s_1) - V_{1,\muts,\nu}^{\met}(s_1) \right] \right| + \left|\ex_k\left[V_{1,\muts,\nu}^{\met}(s_1) - V_{1,\muts,\nu}^{\E}(s_1) \right] \right| \right)^2}{\I_k^{\muts,\nu}(\E; \T^{k}_{H+1})} \\
    &\le \frac{\left(\sqrt{SABH^3 \I_k^{\muts,\nu}(\E; \T^{k}_{H+1})} + \sqrt{\frac{H^3}{2} \I_k^{\muts,\nu}(\E; \T^{k}_{H+1}) } \right)^2}{\I_k^{\muts,\nu}(\E; \T^{k}_{H+1})} \\
    &= 4SABH^3.
\end{align}
This completes the proof of Lemma~\ref{lemma:ts}.

\section{Appendix for General-sum MGs: Proofs of Theorem~\ref{thm:general}} \label{appendix:general}

In the following, we focus on the NE setting where the policy is $\pids^{\text{NE}} = \{\pidskne \}_{k\in[K]}$. We remark that the CCE setting can be analyzed in the same manner by simply replacing product policies to joint policies.

We first present a lemma (analogous to Lemma~\ref{lemma:eq} for zero-sum MGs) that is crucial for the subsequent proofs. The proof of is omitted since it is similar to that of Lemma~\ref{lemma:eq}. 
\begin{lemma} \label{lemma:eq2}
    For any joint policy $\pi$, we have  
$$V_{\pi}^{(i),\meh}(s_1) = \ex_k\left[ V_{\pi}^{(i),\E}(s_1) \right] + \lambda \I_k^{\pi}(\E; \T^k_{H+1}).$$
\end{lemma}

For notational convenience, we abbreviate $\pids^{\text{NE}} = \{\pidskne \}_{k\in[K]}$ as $\pi = \{\pi^k \}_{k \in [K]}$, and abbreviate $V^{(i),\E}_{1,\pi}(s)$ as $V^{(i),\E}_{\pi}(s)$ when $h = 1$. Moreover, we let $\mu^{(i),\pi}_{\E} \triangleq (\pi^{(i),\dagger}_{\E}, \pi^{(-i)})$. 

\paragraph{Proof of regret bounds}
Recall that the Bayesian NE regret of $\pi$ takes the form
\begin{align}
&\mathsf{BR}^{\text{NE}}_K (\pi) \\
&= \ex_{\E \sim \rho} \left(\reg^{\text{NE}}_K(\E,\pi) \right) \\ 
&= \ex\left(\sum_{k=1}^K \sum_{i=1}^{\N} V^{(i),\E}_{\mu_{\E}^{(i),\pi^k}}(s_1) - V^{(i),\E}_{\pi^k}(s_1) \right), \\
&= \sum_{k=1}^K \ex_{\D_k} \left[ \sum_{i=1}^{\N} \ex_k \left(  V^{(i),\E}_{\mu_{\E}^{(i),\pi^k}}(s_1) - V^{(i),\E}_{\pi^k}(s_1) \right) \right]
 \end{align}
where $\mu_{\E}^{(i),\pi^k} = (\pi^{k,(i),\dagger}_{\E},  \pi^{k,(-i)})$. For any fixed $\D_k$ and any player $i \in [\N]$, we have 
\begin{align}
&\ex_k \left( V^{(i),\E}_{\mu_{\E}^{(i),\pi^k}}(s_1) - V^{(i),\E}_{\pi^k}(s_1) \right) \\
&=\ex_k \left( V^{(i),\E}_{\mu_{\E}^{(i),\pi^k}}(s_1)\right) - \left[ \ex_k \left(V^{(i),\E}_{\pi^k}(s_1) \right) + \lambda \I_k^{\pi^k}(\E; \T_{H+1}^k)  \right] + \lambda \I_k^{\pi^k}(\E; \T_{H+1}^k) \\
&= \ex_k \left( V^{(i),\E}_{\mu_{\E}^{(i),\pi^k}}(s_1)\right) - V_{\pi^k}^{(i),\meh}(s_1) + \lambda \I_k^{\pi^k}(\E; \T_{H+1}^k), \label{eq:general_0}
\end{align}
where the last step follows from Lemma~\ref{lemma:eq2}.
As $\pi^k$ can be a random policy, we sometimes write $V_{\pi^k}^{(i),\meh}(s_1)$ as $\ex_{\nu \sim \pi^k} [V_{\nu}^{(i),\meh}(s_1)]$ for clarity. For player $i$, we introduce the Thompson sampling (TS) policy $\pits$. The TS policy $\pits$ first samples a realization of the environment $\E = e$ according to the distribution $\E \sim \PP(\cdot|\D_k)$, and
then chooses the best response $\pi_{e}^{k,(i),\dagger}$ with respect to $\pi^{k,(-i)}$ in the environment $e$. Note that each best response $\pi_{e}^{k,(i),\dagger}$ is a pure policy, while the TS policy is a random policy. Since $\pi^k$ is a Nash equilibrium of the MG $\meh$, we have 
\begin{align}
V_{\pi^k}^{(i),\meh}(s_1) =  \ex_{\nu \sim \pi^k} [V_{\nu}^{(i),\meh}(s_1)] \ge \ex_{\nu \sim (\pits \times \pi^{k,(-i)})} [V_{\nu}^{(i),\meh}(s_1)].  \label{eq:general_1} 
\end{align}
According to Lemma~\ref{lemma:eq2}, we have 
\begin{align}
\ex_{\nu \sim (\pits \times \pi^{k,(-i)})} [V_{\nu}^{(i),\meh}(s_1)] =  \ex_{\nu \sim (\pits \times \pi^{k,(-i)})} \left[ \ex_k\left[ V_{\nu}^{(i),\E}(s_1) \right] + \lambda \I_k^{\nu}(\E; \T^k_{H+1}) \right]. \label{eq:general_2}
\end{align}
Substituting~\eqref{eq:general_1}-\eqref{eq:general_2} to~\eqref{eq:general_0}, we have 
\begin{align}
&\ex_k \left( V^{(i),\E}_{\mu_{\E}^{(i),\pi^k}}(s_1) - V^{(i),\E}_{\pi^k}(s_1) \right) \\
&\le \ex_k \left( V^{(i),\E}_{\mu_{\E}^{(i),\pi^k}}(s_1)\right) - \left(\ex_{\nu \sim (\pits \times \pi^{k,(-i)})} \left[ \ex_k\left[ V_{\nu}^{(i),\E}(s_1) \right] + \lambda \I_k^{\nu}(\E; \T^k_{H+1}) \right] \right) + \lambda \I_k^{\pi^k}(\E; \T_{H+1}^k) \\
&= \ex_k \left( V^{(i),\E}_{\mu_{\E}^{(i),\pi^k}}(s_1) - \ex_{\nu \sim (\pits \times \pi^{k,(-i)})} \left[ V_{\nu}^{(i),\E}(s_1) \right] \right) - \lambda \I_k^{\pits \times \pi^{k,(-i)}}(\E; \T_{H+1}^k) + \lambda \I_k^{\pi^k} (\E; \T_{H+1}^k) \\
&= \frac{\ex_k \left( V^{(i),\E}_{\mu_{\E}^{(i),\pi^k}}(s_1) - \ex_{\nu \sim (\pits \times \pi^{k,(-i)})} \left[ V_{\nu}^{(i),\E}(s_1) \right] \right)}{\sqrt{\lambda \I_k^{\pits \times \pi^{k,(-i)}}(\E; \T_{H+1}^k)}} \cdot \sqrt{\lambda \I_k^{\pits \times \pi^{k,(-i)}}(\E; \T_{H+1}^k)} \notag \\
&\qquad\qquad\qquad\qquad\qquad\qquad\qquad\qquad\qquad\qquad - \lambda \I_k^{\pits \times \pi^{k,(-i)}}(\E; \T_{H+1}^k) + \lambda \I_k^{\pi^k} (\E; \T_{H+1}^k) \\
&\le \frac{\left[\ex_k \left( V^{(i),\E}_{\mu_{\E}^{(i),\pi^k}}(s_1) - \ex_{\nu \sim (\pits \times \pi^{k,(-i)})} \left[ V_{\nu}^{(i),\E}(s_1) \right] \right)\right]^2}{4\lambda \I_k^{\pits \times \pi^{k,(-i)}}(\E; \T_{H+1}^k)} +\lambda \I_k^{\pi^k} (\E; \T_{H+1}^k), \label{eq:general_4}
\end{align}
where~\eqref{eq:general_4} follows from the AM–GM inequality. Below, we provide an upper bound on the numerator of the first term in~\eqref{eq:general_4}. 

Let's introduce a new environment $$\meh' = (H,\s,\A, \{P_h^{\meh'}\}_{h\in [H]}, \{r_h^{(i)}\}_{h\in[H],i\in[\N]} ),$$ 
where the transition kernel $P_h^{\meh'}(\cdot |s,a) = \ex_{\E \sim \PP(\cdot|\D_k)} [P_h^{\E}(\cdot|s,a)]$, and we point out that $\meh'$ differs from $\meh$ defined in Section~\ref{sec:general} only in terms of the reward functions. Considering~\eqref{eq:general_4}, we have  
\begin{align}
&\left|\ex_k \left( V^{(i),\E}_{\mu_{\E}^{(i),\pi^k}}(s_1) - \ex_{\nu \sim (\pits \times \pi^{k,(-i)})} \left[ V_{\nu}^{(i),\E}(s_1) \right] \right) \right| \\
&\le  \left|\ex_k \left[ V^{(i),\E}_{\mu_{\E}^{(i),\pi^k}}(s_1) \right] - \ex_{\nu \sim (\pits \times \pi^{k,(-i)})} \left[ V_{\nu}^{(i),\meh'}(s_1) \right] \right|  \notag \\
&\qquad\qquad\qquad\qquad\qquad\qquad\qquad + \left|\ex_{\nu \sim (\pits \times \pi^{k,(-i)})} \left[ V_{\nu}^{(i),\meh'}(s_1) \right] - \ex_k\left[\ex_{\nu \sim (\pits \times \pi^{k,(-i)})} \left[ V_{\nu}^{(i),\E}(s_1) \right] \right] \right|. \label{eq:general_5} 
\end{align}
We then consider the two terms in~\eqref{eq:general_5} separately.

For the first term in~\eqref{eq:general_5}, by recalling $\mu_{\E}^{(i),\pi^k} = (\pi^{k,(i),\dagger}_{\E},  \pi^{k,(-i)})$ as well as the definition of the TS policy $\pits$,  we have
\begin{align}
 &\left|\ex_k \left[ V^{(i),\E}_{\mu_{\E}^{(i),\pi^k}}(s_1) \right] - \ex_{\nu \sim (\pits \times \pi^{k,(-i)})} \left[ V_{\nu}^{(i),\meh'}(s_1) \right] \right| \\
 &= \left|\ex_k \left[ \ex_{\nu \sim (\pi^{k,(i),\dagger}_{\E} \times \pi^{k,(-i)}) }
 \left[V^{(i),\E}_{\nu}(s_1) \right] \right] - \ex_k \left[ \ex_{\nu \sim (\pi_{\E}^{k,(i),\dagger} \times \pi^{k,(-i)})} \left[ V_{\nu}^{(i),\meh'}(s_1) \right] \right] \right| \\
 &= \left|\ex_k \left[ \ex_{\nu \sim (\pi^{k,(i),\dagger}_{\E} \times \pi^{k,(-i)} ) } \left( V^{(i),\E}_{\nu}(s_1) - V_{\nu}^{(i),\meh'}(s_1) \right) \right] \right|. \label{eq:95}
\end{align}
Adapting Lemma~\ref{lemma:diff} to the general-sum MG setting, one can obtain that 
\begin{align}
&V^{(i),\E}_{\nu}(s_1) - V_{\nu}^{(i),\meh'}(s_1) = \sum_{h=1}^H \ex_{\nu}^{\meh}\left[ \Delta(\E, \meh', \nu, s_h^k, a_h^k) \right], \\
&\qquad\qquad\qquad\qquad \text{where} \ \Delta(\E, \meh',\nu, s, a) \triangleq \ex_{s' \sim P_h^{\E}(\cdot|s,a)} [V_{h+1,\nu}^{(i),\E}(s')] - \ex_{s' \sim P_h^{\meh'}(\cdot|s,a)} [V_{h+1,\nu}^{(i),\E}(s')].
\end{align}
We also define the occupancy measure with respect to  any joint policy $\pi$ and environment $e$  as
\begin{align}
d_{h,\pi}^{e}(s,a) \triangleq \PP_{\pi}^{e}(s_h^k = s, a_h^k = a).
\end{align}
Therefore,~\eqref{eq:95} can be upper-bounded as
\begin{align}
&\left|\ex_k \left[ \ex_{\nu \sim (\pi^{k,(i),\dagger}_{\E} \times \pi^{k,(-i)}) } \left( V^{(i),\E}_{\nu}(s_1) - V_{\nu}^{(i),\meh'}(s_1) \right) \right] \right| \\
&= \left|\ex_k \left[ \ex_{\nu \sim (\pi^{k,(i),\dagger}_{\E} \times \pi^{k,(-i)})} \left( \sum_{h=1}^H \ex_{\nu}^{\meh'}\left[ \Delta(\E, \meh', \nu, s_h^k, a_h^k) \right] \right) \right] \right| \\
&\le \sum_{h=1}^H \ex_k \ex_{\nu \sim (\pi^{k,(i),\dagger}_{\E} \times \pi^{k,(-i)})} \left[ \sum_{s,a} d_{h,\nu}^{\meh'}(s,a) \cdot |\Delta(\E, \meh', \nu, s, a)| \right] \\
&= \sum_{h=1}^H \ex_k \ex_{\nu \sim (\pi^{k,(i),\dagger}_{\E} \times \pi^{k,(-i)})} \sum_{s,a} \frac{d_{h,\nu}^{\meh'}(s,a)}{\sqrt{\ex_k \ex_{\nu \sim (\pi^{k,(i),\dagger}_{\E} \times \pi^{k,(-i)})} [d_{h,\nu}^{\meh'}(s,a)]}}  \notag\\
&\qquad\qquad\qquad\qquad \qquad\qquad\qquad\qquad \times  \sqrt{\ex_k \ex_{\nu \sim (\pi^{k,(i),\dagger}_{\E} \times \pi^{k,(-i)})} [d_{h,\nu}^{\meh'}(s,a)]} \cdot |\Delta(\E, \meh', \nu, s, a)| \\
&\le \sqrt{\sum_{h=1}^H \ex_k \ex_{\nu \sim (\pi^{k,(i),\dagger}_{\E} \times \pi^{k,(-i)})} \sum_{s,a} \frac{d_{h,\nu}^{\meh'}(s,a)^2}{\ex_k \ex_{\nu \sim (\pi^{k,(i),\dagger}_{\E} \times \pi^{k,(-i)})} [d_{h,\nu}^{\meh'}(s,a)]} }  \notag \\
&\qquad\qquad\qquad\qquad \times \sqrt{\sum_{h=1}^H \ex_k \ex_{\nu \sim (\pi^{k,(i),\dagger}_{\E} \times \pi^{k,(-i)})} \sum_{s,a} (\ex_k \ex_{\nu \sim (\pi^{k,(i),\dagger}_{\E} \times \pi^{k,(-i)})} [d_{h,\nu}^{\meh'}(s,a)]) \Delta(\E, \meh', \nu, s, a)^2}. \label{eq:103} 
 \end{align}
Note that the first term in~\eqref{eq:103} is upper-bounded by $\sqrt{SAH}$ by using the fact that $d_{h,\nu}^{\meh'}(s,a)^2 \le d_{h,\nu}^{\meh'}(s,a)$. For the second term in~\eqref{eq:103}, we note that 
$$\ex_k \ex_{\nu \sim (\pi^{k,(i),\dagger}_{\E} \times \pi^{k,(-i)})} [d_{h,\nu}^{\meh'}(s,a)] = d^{\meh'}_{h,(\pits \times \pi^{k,(-i)})}(s,a)$$ by the definition of the TS policy $\pits$. Moreover, by following the steps in~\eqref{eq:old1}-\eqref{eq:tan3} for zero-sum MGs, one can upper-bound the second term in~\eqref{eq:103} by
\begin{align}
    \sqrt{\frac{H^2}{2} \sum_{h=1}^H \ex_k \left[ \sum_{s,a} d^{\meh'}_{h,(\pits \times \pi^{k,(-i)})}(s,a) \cdot  \DD_{\text{KL}}\left(P_h^{\E}(\cdot|s,a) \Vert P_h^{\meh'}(\cdot|s,a)\right)  \right] }.
\end{align}
Analogous to Lemma~\ref{lemma:mikl}, one can also show that in the context of general-sum MGs,  
\begin{align}
   \sum_{h=1}^H \ex_k \left[ \sum_{s,a} d^{\meh'}_{h,(\pits \times \pi^{k,(-i)})}(s,a) \cdot  \DD_{\text{KL}}\left(P_h^{\E}(\cdot|s,a) \Vert P_h^{\meh'}(\cdot|s,a)\right)  \right] = \I_k^{(\pits \times \pi^{k,(-i)})}(\E; \T_{H+1}^k). 
\end{align}
Thus, the first term in~\eqref{eq:general_5} satisfies  
\begin{align}
    &\left|\ex_k \left[ V^{(i),\E}_{\mu_{\E}^{(i),\pi^k}}(s_1) \right] - \ex_{\nu \sim (\pits \times \pi^{k,(-i)})} \left[ V_{\nu}^{(i),\meh'}(s_1) \right] \right| \notag \\
    &\qquad\qquad\qquad\qquad\qquad\le \sqrt{\frac{1}{2}SAH^3 \cdot \I_k^{(\pits \times \pi^{k,(-i)})}(\E; \T_{H+1}^k) }.
\end{align}

Next, we consider the second term in~\eqref{eq:general_5}. Note that 
\begin{align}
&\left|\ex_{\nu \sim (\pits \times \pi^{k,(-i)})} \left[ V_{\nu}^{(i),\meh'}(s_1) \right] - \ex_k\left[\ex_{\nu \sim (\pits \times \pi^{k,(-i)})} \left[ V_{\nu}^{(i),\E}(s_1) \right] \right] \right| \\
&=\left|\ex_k \left[ \sum_{h=1}^H \ex_{\nu \sim (\pits \times \pi^{k,(-i)})}  \ex_{\nu}^{\meh'}\left[ \sum_{s'} \left[P_{h}^{\meh'}(s'|s_h^k, a_h^k) - P_{h}^{\E}(s'|s_h^k, a_h^k)  \right] \cdot V_{\nu}^{(i),\E}(s') \right] \right] \right| \\
&\le H \cdot \sum_{h=1}^H \ex_k\left[ \sum_{s,a} d_{h, (\pits \times \pi^{k,(-i)})}^{\meh'} (s,a) \cdot \sqrt{\frac{1}{2}\DD_{\text{KL}}\left(P_h^{\E}(\cdot|s,a) \Vert P_h^{\meh'}(\cdot|s,a)\right)} \right] \\
&\le H^2 \sqrt{ \frac{1}{2H} \sum_{h=1}^H \ex_k \left[\sum_{s,a} d_{h, (\pits \times \pi^{k,(-i)})}^{\meh'} (s,a) \cdot \DD_{\text{KL}}\left(P_h^{\E}(\cdot|s,a) \Vert P_h^{\meh'}(\cdot|s,a)\right) \right] } \label{eq:jen} \\
&\le \sqrt{\frac{H^3}{2} \I_k^{(\pits \times \pi^{k,(-i)})}(\E; \T_{H+1}^k)},
\end{align}
where~\eqref{eq:jen} follows from Jensen's inequality. 

Therefore, we have 
\begin{align}
    \left|\ex_k \left( V^{(i),\E}_{\mu_{\E}^{(i),\pi^k}}(s_1) - \ex_{\nu \sim (\pits \times \pi^{k,(-i)})} \left[ V_{\nu}^{(i),\E}(s_1) \right] \right) \right| \le 2\sqrt{SAH^3 \I_k^{(\pits \times \pi^{k,(-i)})}(\E; \T_{H+1}^k)}. \label{eq:sub}
\end{align}
Substituting~\eqref{eq:sub} to~\eqref{eq:general_4} yields that 
\begin{align}
    \ex_k \left( V^{(i),\E}_{\mu_{\E}^{(i),\pi^k}}(s_1) - V^{(i),\E}_{\pi^k}(s_1) \right) \le \frac{SAH^3}{\lambda} + \lambda \cdot \I_k^{\pi^k}(\E; \T_{H+1}^k).
\end{align}
Finally, we obtain that
\begin{align}
\mathsf{BR}^{\text{NE}}_K (\pi) &= \sum_{k=1}^K \ex_{\D_k} \left[ \sum_{i=1}^{\N} \ex_k \left(  V^{(i),\E}_{\mu_{\E}^{(i),\pi^k}}(s_1) - V^{(i),\E}_{\pi^k}(s_1) \right) \right] \\
&= \frac{SAH^3 KN}{\lambda} + \lambda N \sum_{k=1}^K \ex_{\D_k} \left[\I^{\pi^k}_k(\E; \T_{H+1}^k) \right] \\
&= \frac{SAH^3 KN}{\lambda} + \lambda N \cdot \I^{\pi}(\E; \D_{K+1}) \\
&= 3NS^{3/2}AH^2 \sqrt{K \log(SKH)},
 \end{align}
 where the last step is obtained by setting $\lambda = \sqrt{HK^2/S\log(SKH)}$. This completes the proof of Theorem~\ref{thm:general}.

\section{Proof of Lemma~\ref{lemma:property}}
\begin{proof} \label{appendix:lemma2}

Let $\Pi_A = \{\s \to \Delta(\A) \}^H$ and $\Pi_B = \{\s \to \Delta(\B) \}^H$. 
For any subspace $\Theta_c$, any pair of environment $e,e' \in \Theta_c$, and any policy $(\mu,\nu) \in \Pi_A \times \Pi_B$, we have
\begin{align}
&V_{1,\mu,\nu}^e(s_1) - V_{1,\mu,\nu}^{e'}(s_1) \\
&=\sum_{h=1}^H \ex_{\mu,\nu}^{e'}\left[ \ex_{s' \sim P_h^e(\cdot|s_h,a_h,b_h)} [V_{h+1,\mu,\nu}^e(s')] - \ex_{s' \sim P_h^{e'}(\cdot|s_h,a_h,b_h)} [V_{h+1,\mu,\nu}^e(s') \right]  \label{eq:sss}\\
&=\sum_{h=1}^H \ex_{\mu,\nu}^{e'} \left[[P_h^{e}(\cdot|s_h,a_h,b_h) - P_h^{e'}(\cdot|s_h,a_h,b_h)] \cdot V_{h+1,\mu,\nu}^e(s')  \right] \\
&\le H \sum_{h=1}^H \ex_{\mu,\nu}^{e'} \left[ \big|P_h^{e}(\cdot|s_h,a_h,b_h) - P_h^{e'}(\cdot|s_h,a_h,b_h)\big| \right] \\
&\le H \sum_{h=1}^H \max_{s,a,b} \big|P_h^{e}(\cdot|s,a,b) - P_h^{e'}(\cdot|s,a,b)\big| \\
&\le \epsilon,
\end{align}
where~\eqref{eq:sss} follows from Lemma~\ref{lemma:diff}, and the last inequality follows from the property of $\Theta_c$. Therefore, based on the construction of $\te$ in~\eqref{eq:concrete}, it is clear that 
\begin{align}
        \ex_k[d_{\Pi_A, \Pi_B}(\E,\te)] \le \epsilon \quad \text{and} \quad  \PP\left(d_{\Pi_A, \Pi_B}(\E,\te) > \epsilon \right) = 0.
    \end{align}
This completes the proof of Lemma~\ref{lemma:property}.
\end{proof}

\bibliographystyle{IEEEtran}
\bibliography{biblio}
\end{document}